\newcount\Comments  
\documentclass[11pt]{article}
\usepackage[margin=1in]{geometry}
\usepackage{graphicx} 
\usepackage{booktabs} 
\usepackage[ruled]{algorithm2e} 

\SetAlFnt{\small}
\SetAlCapFnt{\small}
\SetAlCapNameFnt{\small}
\SetAlCapHSkip{0pt}
\IncMargin{-\parindent}

\usepackage{graphicx} 
\usepackage[utf8]{inputenc}

\usepackage{color-edits}
\addauthor{sw}{blue}
\usepackage{amsmath,amssymb,amsfonts,amsthm}
\usepackage{hyperref}
\usepackage{caption}
\usepackage{subcaption}
\expandafter\def\csname ver@subfig.sty\endcsname{}
\usepackage{enumerate}
\usepackage{paralist}
\renewenvironment{enumerate}[1]{\begin{compactenum}#1}{\end{compactenum}}

\usepackage{url}
\usepackage{bbm}
\usepackage[nameinlink,capitalise]{cleveref} 
\crefname{ineq}{Inequality}{Inequalities}
\usepackage{subfig}
\usepackage{multirow}
\usepackage{paralist}

\newcommand{\ignore}[1]{}
\newcommand{\kibitz}[2]{\ifnum\Comments=1\textcolor{#1}{#2}\fi}
\newcommand{\shuran}[1]{\kibitz{cyan}{[Shuran: #1]}}
\newcommand{\jiahao}[1]{\kibitz{teal}{[Jiahao: #1]}}

\newtheorem{definition}{Definition}[section]
\newtheorem{lemma}{Lemma}[section]
\newtheorem{theorem}{Theorem}[section]
\newtheorem{proposition}{Proposition}[section]

\newtheorem{example}{Example}[section]
\newtheorem{remark}{Remark}[section]
\newtheorem{assumption}{Assumption}[section]
\newtheorem{fact}{Fact}[section]
\newtheorem{corollary}{Corollary}[section]

\def\+#1{\mathcal{#1}}
\def\-#1{\mathbb{#1}}

\newcommand{\notshow}[1]{{}}
\newcommand{\AutoAdjust}[3]{{ \mathchoice{ \left #1 #2  \right #3}{#1 #2 #3}{#1 #2 #3}{#1 #2 #3} }}
\newcommand{\Xcomment}[1]{{}}

\newcommand{\InBrackets}[1]{\AutoAdjust{[}{#1}{]}}

\newcommand{\InAbs}[1]{\AutoAdjust{|}{#1}{|}}
\renewcommand{\part}[2]{\frac{\partial #1}{\partial #2}}

\DeclareMathOperator*{\argmax}{arg\,max}

\usepackage{authblk}

\usepackage{natbib}
\setcitestyle{authoryear,open={[},close={]}}

\title{Ex-post Individually Rational Bayesian Persuasion}
\author[1]{Jiahao Zhang}
\author[2]{Shuran Zheng}
\author[3]{Renato Paes Leme}
\author[1]{Zhiwei Steven Wu}
\affil[1]{Carnegie Mellon University}
\affil[2]{Tsinghua University}
\affil[3]{Google Research}

\begin{document}
\maketitle
\begin{abstract}
In the Bayesian persuasion model, a sender can convince a receiver to choose an alternative action to the one originally preferred by the receiver. A crucial assumption in this model is the sender's commitment to a predetermined information disclosure policy (signaling scheme) and the receiver's trust in this commitment.  However, in practice, it is difficult to monitor whether the sender adheres to the disclosure policy, and the receiver may refuse to follow the persuasion due to a lack of trust. Trust becomes particularly strained when the receiver knows that the sender will incur obvious losses when truthfully following the protocol.  In this work, we propose the notion of \emph{ex-post individually rational (ex-post IR)} Bayesian persuasion: after observing the state, the sender is never asked to send a signal that is less preferred than no information disclosure. An ex-post IR Bayesian persuasion policy is more likely to be truthfully followed by the sender, thereby providing stronger incentives for the receiver to trust the sender. Our contributions are threefold. First, we demonstrate that the optimal ex-post IR persuasion policy can be efficiently computed through a linear program, while also offering its geometric characterization. Second, we show that surprisingly, for non-trivial classes of games, the requirement of ex-post IR constraints does not incur any cost to the sender's utility. Finally, we compare ex-post IR Bayesian persuasion to other information disclosure models that ensure different notions of credibility.
\end{abstract}

\section{Introduction}
In the wake of the seminal work by \citet{rayo2010optimal} and \citet{kamenica2011bayesian},
there has been a significant surge in research interests over the past decade concerning the design of optimal information disclosure policies, especially Bayesian persuasion. In the standard formulation of Bayesian persuasion, a sender strategically discloses information to influence a receiver's belief and chosen action. A crucial assumption in this model is that the sender will commit to a signaling policy and the receiver will trust this commitment. However, in most practical scenarios, this commitment is difficult to monitor, which may lead to the receiver refusing to follow the persuasion due to a lack of trust. The receiver's trust is especially fragile when there are clear deviations that could benefit the sender in certain situations. Consider the following scenario: 

\begin{example}[Lending]\label{ex:lending}
    There is a credit reporting agency (the sender) recommending borrowers to a lender (the receiver). The lender makes loaning decisions based on the borrowers' trustworthiness, i.e., whether the loan will be repaid. Suppose that the lender will reject the loan application if the probability of repayment is below $0.3$. The lender will issue a small loan if the probability is in $[0.3,1)$ and grant a huge loan if the probability is $1$. Suppose the lender's prior about a borrower's repayment probability is $0.5$, which means that without any recommendation from the sender, the lender will issue a small loan. Suppose the sender knows whether a borrower will repay the loan based on their credit history, and the sender can send a signal to the lender. The sender's goal is to persuade the lender to approve loans. The sender's utility is $0$ if the loan is rejected, $1$ if a small loan is approved, and $10$ if a huge loan is approved.  Then the optimal Bayesian persuasion policy that maximizes the sender's expected utility is to fully reveal whether the borrower will repay the loan, which gives an expected utility of $5$ supposing the prior probability of repayment is $0.5$.
\end{example}

The optimal information disclosure policy in Example \ref{ex:lending} has its Achilles' heel. The sender may hesitate to truthfully reveal low-credit borrowers who will not repay the loans, as it will yield a lower utility than no information disclosure, which guarantees a small loan for the borrower. Consequently, the receiver might cast doubt on the reliability of the entire protocol, as the sender incurs a clear loss when truthfully sending the signal. In the absence of trust, this policy may fail to effectively persuade the receiver. Motivated by this tension between incentives and trust, our work considers the notion of \emph{ex-post individually rational } (ex-post IR) Bayesian persuasion: after observing the state (the borrower's trustworthiness), the sender never sends a signal that yields a lower payoff for the sender than no information disclosure. We provide three types of results.




\begin{enumerate}

\item First, we study how to compute the optimal signaling scheme for ex-post IR Bayesian persuasion. We show that the optimal signaling scheme can still be solved by linear programming even if we add the ex-post IR constraint. In addition, we provide a geometric characterization of the optimal ex-post IR signaling scheme by extending the concave closure representation by \citet{kamenica2011bayesian} to incorporate the ex-post IR constraint. In the binary state setting, our geometric result enables a clean visualization of the optimal ex-post IR signaling policy.


\item Second, we establish conditions where the ex-post IR constraint imposes no cost to the sender. While there are instances in which the ex-post IR constraint incurs a cost to the optimal sender utility, a natural question is when ex-post IR incurs no cost. Building on our aforementioned geometric result, we give a complete geometric characterization in the binary state case (\cref{thm:n*2}) of when the sender achieves the same optimal utilities whether or not ex-post IR constraint is required. Our proof technique links the optimal ex-post IR signaling scheme with the \emph{quasiconcave closure} of the sender's expected utility function. As a corollary, the geometric characterization enables a more computationally efficient algorithm than simply running linear programming to determine whether the ex-post IR constraint incurs a sender utility gap. 

However, our geometric approach for the binary state setting does not extend to the general state space setting. Going beyond binary states, we focus on characterizing the utility gap for two families of Bayesian persuasion problems that are motivated by real-world applications.




\begin{itemize}
    \item We define a broad class of problems called \emph{trading games} (\cref{def:trading}), which include \emph{bilateral trade} and first-price auction with reserve price as special cases. We show that the optimal signaling schemes of all trading games are ex-post IR. 
    \item We identify a general class of games that meet two conditions called \emph{cyclical monotonicity} and \emph{weak logarithmic supermodularity}, with the \emph{credence goods} market as a notable special case. We introduce a new notion of the \emph{greedy signaling scheme} and show that it is ex-post IR for this class. Then in the special case of the credence goods market, we show that there exists a greedy signaling scheme is optimal, which implies that ex-post IR constraint incurs no cost to the sender's utility.   
    
\end{itemize}

\item Third, we compare the sender's optimal expected utility under varying degrees of commitment power. At one extreme is the standard Bayesian persuasion model, where the sender has full commitment power, i.e., they can commit any Bayes-plausible signaling scheme as they like. At the other extreme is the cheap talk model, where the sender has no commitment power at all. As a result, the sender's optimal expected utility is always higher in Bayesian persuasion than in cheap talk (\cite{lipnowski2020equivalence}).  We show that in the ex-post IR Bayesian persuasion model, the sender's optimal utility is always between the optimal utilities of Bayesian persuasion and cheap talk \emph{if and only if} the sender has an ordered preference over actions (see \cref{sorted}). Finally, we compare with a model called \emph{credible persuasion} studied by \cite{lin2022credible}, which also studied the credibility of the sender. We show that the sender's optimal utility in credible persuasion is also between the optimal utilities of Bayesian persuasion and cheap talk. However, we provide two examples showing that the utility of the optimal ex-post IR signaling scheme could be higher or lower than that of optimal credible persuasion.
\end{enumerate}

The rest of the paper is organized as follows: \cref{section:model} introduces the model of Bayesian persuasion and the definition of ex-post IR Bayesian persuasion. \cref{section:optimality} provides results about the optimality of ex-post IR Bayesian persuasion in both algebraic and geometric views, which is similar to \cite{kamenica2011bayesian}. \cref{section:gap} provides results of when there is no sender utility gap between the Bayesian persuasion and the ex-post IR Bayesian persuasion. \cref{section:compare} compares the optimal sender's utilities under different information transmission models that ensure different notions of credibility. The related literature and all omitted proofs are in the appendix. 

\section{Model}\label{section:model}
We introduce the standard Bayesian persuasion problem with finite actions and states, where a sender chooses a signal to reveal information to a receiver who then takes an action. The sender's utility $v(a, \theta)$ and the receiver's utility $u(a, \theta)$ depend on the receiver's action $a\in A$ \shuran{Why do we need the subscript $i$ here? It seems to me $a\in A$ is fine.} and the state of world $\theta\in\Theta$. Throughout this paper, we impose two standard assumptions. First, the sender has\shuran{Reworded a bit here.} a preference over the actions regardless of the state.  See \cref{sorted} for the formal description. Second, there is no dominated action for the receiver, i.e. for any action $a$, there exists a posterior belief $\mu$ that $a$\shuran{Again, it might be better to just use $a$} is the best action for the receiver under $\mu$. This is without loss of generality because we can drop dominated actions.


\begin{assumption}\label{sorted}
    Let $\InAbs{A}=n$. The sender's preference over action is $a_1\ge a_2\ge...\ge a_n$. In other words, for any $\theta\in\Theta$ and any $i,j\in[n]$ that $i\le j$, we have
$v(a_i,\theta)\ge v(a_j,\theta)$.
\end{assumption}
The sender and receiver share a common prior $\mu_0\in\Delta(\Theta)$ about the state. The sender commits to a \emph{signaling scheme} $\phi$, which is a randomized mapping from the state to signals. Formally, the signaling scheme is defined as $\phi: \Theta\mapsto\Delta(S)$, where $S$ is the set of all possible signals. For convenience, let $\phi_{\theta} \in \Delta(S)$ be the conditional signal distribution when the state is $\theta\in\Theta$\shuran{Reworded a bit here} and we denote by $\phi_{\theta}(s)$ the probability of sending signal $s\in S$ when the state is $\theta$. For ease of notation, we use $\phi$ to also denote the distribution over signals induced by the signaling scheme $\phi$ and the prior $\mu_0$ and denote $\phi(s)$ the probability of sending signal $s\in S$.


The interaction between the sender and the receiver goes on as follows:
\begin{enumerate}
    \item the sender commits to a signaling scheme $\phi$ and the receiver observes it;
    \item the sender observes a realized state of nature $\theta\sim\mu_0$;
    \item the sender draws a signal $s\in S$ according to the distribution $\phi_{\theta}$ and sends it to the receiver;
    \item the receiver observes the signal $s$ and updates her belief about the state according to the \emph{Bayes rule};
    \item the receiver chooses an action that maximizes her expected utility.
\end{enumerate}

The receiver's posterior belief after seeing $s\in S$ is denoted by $\mu_s$. We have $\mu_s(\theta)=\frac{\mu_0(\theta)\phi_\theta(s)}{\phi(s)}$.

For any posterior belief $\mu_s$, denote the receiver's best response $a^{*}(\mu_s)\in\argmax_{a\in A}\sum_{\theta\in\Theta}u(a,\theta)\mu_s(\theta)$. For convenience, let the index of the best response under prior $a^*(\mu_0)$ be $k^*$, i.e. $a^*(\mu_0)=a_{k^*}$. As is customary in the literature, we assume that the receiver breaks ties in favor of the sender. The sender's expected utility when sending $s$ will be equal to $\hat{v}(\mu_s):=\sum_{\theta\in\Theta}v\big(a^*(\mu_s),\theta\big)\mu_s(\theta)$.  An outcome is a pair of the realized world state and the receiver's action. Denote $\pi\in \Delta( A\times\Theta)$\shuran{Changed the expression a bit here} the probability distribution of outcomes given the signaling scheme is $\phi$ and the prior $\mu_0$. We have
\begin{eqnarray*}
  \pi(a,\theta)=\sum_{s\in S:a^*(\mu_s)=a}\sum_{\theta\in\Theta}\phi_{\theta}(s)\mu_0(\theta).  
\end{eqnarray*}.

Sometimes we use $\pi$ to represent a signaling scheme. Then, we introduce the definition of the ex-post individually rational (IR) signaling scheme.
\begin{definition}[Ex-post individual rationality]
A signaling scheme $\phi$ is ex-post IR for the sender if for any state $\theta\in\Theta$ and any signal $s\in S$, $v(a^*(\mu_s),\theta) \ge v(a^*(\mu_0),\theta)$.
 
\end{definition}
The definition of ex-post IR  requires that compared to no communication, the sender will not regret sending the signal $s$ in any case. It can also mean that after the sender sees the state and draws a signal in Step 3, he will not hesitate to send the signal. Next, we give the definition of the Bayesian persuasion instance or \emph{game} for short.{\jiahao{added the definition of game}}
\begin{definition}[game]
    A game (Bayesian persuasion instance) is a quadruple $(A,\Theta,v,u)$ where A is the action set, $\Theta$ is the state set, $v$ is the sender's utility function and $u$ is the receiver's utility function.
\end{definition}
\section{Warmup: optimal ex-post IR signaling schemes\shuran{Changed the title}}\label{section:optimality}
It is well known that the optimal signaling scheme of a standard Bayesian persuasion problem with a finite number of actions and states can be solved by linear programming. In addition, the optimal expected Sender utility can be characterized by the concave closure of the expected Sender utility function $\hat{v}:\Delta(\theta)\mapsto\-R$, which maps a posterior belief to the expected Sender utility. In this section, we extend these two standard results to the setting of ex-post IR Bayesian persuasion. We first show that the ex-post IR constraints are linear, which implies that the optimal ex-post IR signaling scheme can also be solved by linear programming. Then we 
provide a concavefication characterization of the optimal ex-post IR signaling scheme.

\subsection{Linear programming}
The linear programming for Bayesian persuasion is 
\[
    \max_{\pi}\sum_{a\in A}\sum_{\theta\in\Theta}v(a,\theta)\pi(a,\theta)
\]
    subject to
    \begin{equation}\label{eq: incentive}
        \sum_{\theta\in\Theta}(u(a,\theta)-u(a^\prime,\theta))\pi(a,\theta)\le0\quad\forall a,a^\prime\in A
    \end{equation}
    \begin{equation}\label{eq: feasible}
        \sum_{a\in A}\pi(a,\theta)=\mu_0(\theta)\quad\forall \theta\in\Theta,
    \end{equation}
where \cref{eq: incentive} is the incentive compatible constraints for the Receiver and \cref{eq: feasible} is to make sure the signaling scheme is \emph{Bayes-plausible}. To solve the optimal ex-post IR signaling scheme, we only need to add the following ex-post IR constraints
\begin{equation}
    \pi(a,\theta)=0 \quad\forall a\in A,\theta\in\Theta, u(a,\theta)<u(a^*(\mu_0),\theta).
\end{equation}

\subsection{Concave closure}

Let $V$ be the concave closure of the expected utility function $\hat{v}$.
\begin{fact}
    The value of the optimal signaling scheme is $V(\mu_0)$. The sender benefits from persuasion if and only if $V(\mu_0)>\hat{v}(\mu_0)$
\end{fact}

Now we consider Bayesian persuasion with ex-post IR constraints. Instead of computing the concave closure on the whole domain of posterior belief, we define a closure only on the posterior belief $\mu$ that $a^*(\mu)\ge a^*(\mu_0)$. Let $A^+$ be the set of actions preferred by Sender over $a^*(\mu_0)$. Formally we have
\begin{equation}\label{eq:closure}
   V_{\textsc{ex-post}}(\mu_0)=\max_{\tau\in\Delta(\Delta(\Theta))}\left\{\sum_{\mu\in\mathrm{supp}(\tau)}\tau(\mu)\hat{v}(\mu): \sum_{\mu\in\mathrm{supp}(\tau)}\tau(\mu)\mu=\mu_0,\quad a^*(\mu)\in A^+,\forall\mu\in\mathrm{supp}(\tau)\right\}.
\end{equation}

\begin{theorem}
       The value of the optimal ex-post IR signaling scheme is $V_{\textsc{ex-post}}(\mu_0)$. The sender benefits from ex-post IR persuasion if and only if $V_{\textsc{ex-post}}(\mu_0)>\hat{v}(\mu_0)$.
\end{theorem}

The proof directly follows from the \cref{eq:closure} and \cite{kamenica2011bayesian}, which is straightforward and omitted. Recall \cref{ex:lending} in the introduction, \cref{fig:ex-post closure} visualize the concave closure $V$ and the concave closure with ex-post IR constraints $V_{\textsc{ex-post}}$ in this example.

\begin{figure}
    \centering
    \includegraphics[width=0.33\textwidth]{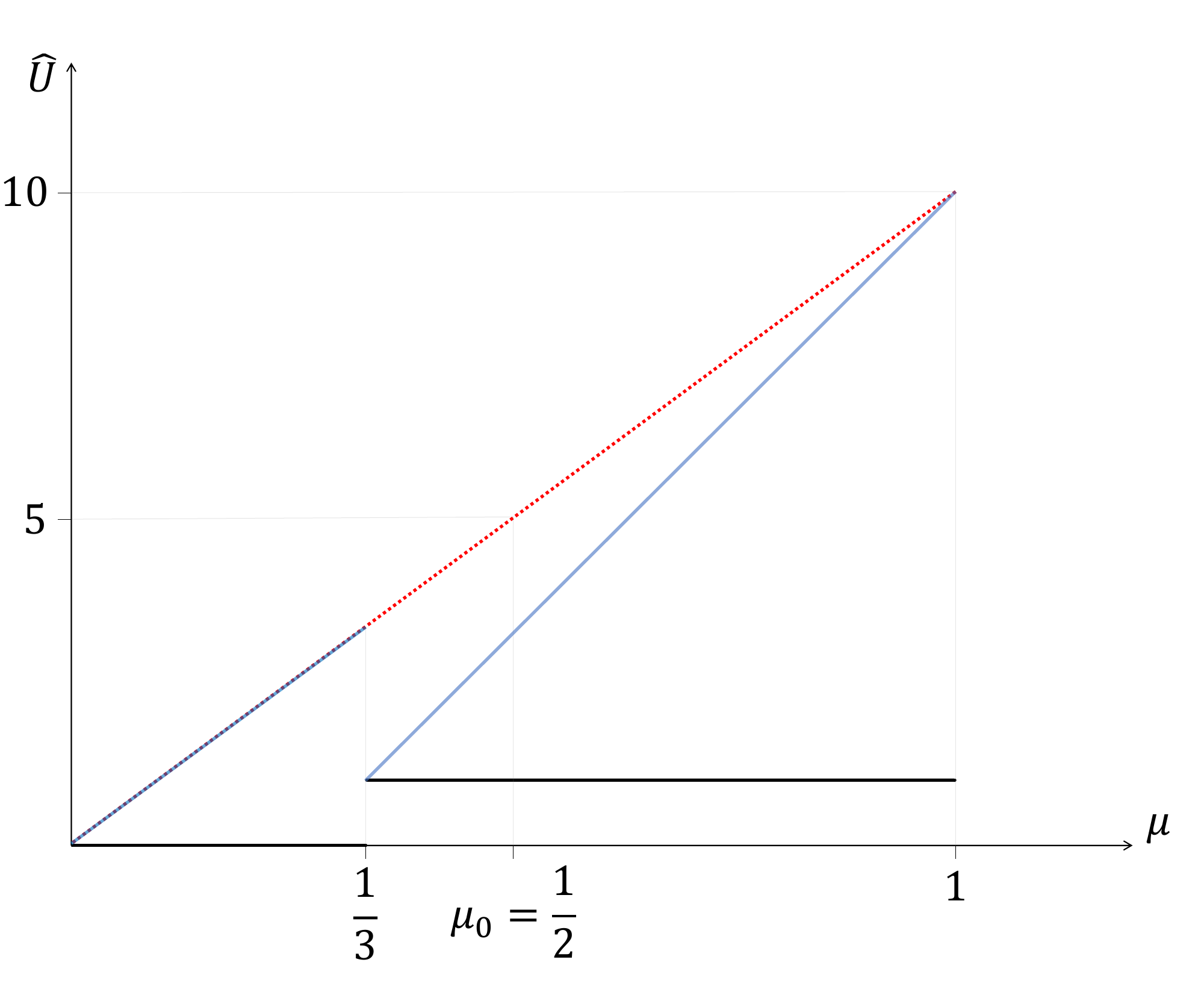}
    \caption{The black line is the utility curve of the sender's expected utility function $\hat{v}$ in \cref{ex:lending}. The red line is the concave closure $V$. The blue line is the concave closure with ex-post IR constraints $V_{\textsc{ex-post}}$.}
    \label{fig:ex-post closure}
\end{figure}
\section{The effect of the ex-post IR constraint\shuran{Changed the title}}
\label{section:gap}
Now that we know how to compute the optimal ex-post IR signaling scheme and its concave closure representation. %
In this section, we compare the sender's optimal signaling schemes with and without the ex-post IR constraint, and examine how the ex-post IR requirement affects the sender's expected utility and provide more explicit results. We focus on the following question: under what conditions does the ex-post IR constraint incurs no cost to the sender's utility?



Firstly, in \cref{subsection:binary}, we characterize the impact of the ex-post IR constraint by a simple function: \emph{smoothed quasiconcave closure}. We show that the ex-post IR constraint imposes no costs if and only if the smoothed quasiconcave closure is concave.
This function is simple: it is basically a refined \emph{quasiconcave closure} of the sender's expected utility function, smoothed by connecting the adjacent endpoints. 
See \cref{thm:n*2} for details.
For the general cases, our results are driven by applications. We identify two non-trivial classes of games with real-world applications, demonstrating that ex-post IR constraint incurs no costs in these games. The first class of games is \emph{trading game} in \cref{subsection:bilateral}. The second class of games is in \cref{subsection:credence} with credence goods market as its application.

\shuran{Moved the detailed discussions to later sections.}

\subsection{Geometric characterization in the binary state case}\label{subsection:binary}
In this section, we consider a binary state set $\Theta=\{\theta_1,\theta_2\}$ and a finite action set $A=\{a_1,a_2,...a_n\}$. In this case, it's known that the optimal signaling scheme can be visualized from the concave closure of the sender's expected utility function. In our paper, instead of considering the concave closure, we show how to determine whether there is a gap between the optimal signaling scheme and optimal ex-post IR signaling scheme through the quasiconcave closure of the sender's expected utility function. For convenience, we first define the \emph{partition} of $[0,1]$ and the thresholds of a partition.
\begin{definition}[partition]
    A family of intervals $P$ is a partition of $[0,1]$ if and only if
    \begin{enumerate}
        \item the family $P$ does not contain $\emptyset$,
        \item the union of intervals in $P$ is equal to $[0,1]$,
        \item the intersection of any two distinct intervals is $\emptyset$.
    \end{enumerate}
For a partition $\{I_i\}_{i=1}^n$ of $[0,1]$, let $\alpha_0=0$ and $\alpha_i$ be the right endpoint of interval $I_i$. We can always sort $\{I_i\}_{i=1}^n$ such that $0=\alpha_0<\alpha_1<\cdots<\alpha_n=1$. The threshold of $I_i$ and $I_{i+1}$ is $\alpha_i$ for any $1\le i<n$. Putting it all together, we denote a partition $P$ with its thresholds as $P=(\{I_i\}_{i=1}^n,\{\alpha_i\}_{i=0}^n)$.
Sometimes we will omit thresholds $\{\alpha_i\}_{i=0}^n$ if they are not needed.
\end{definition}

{\jiahao{In this section, uniformly use "intervals" to represent the subdomains of piecewise function.}}In the binary state case, say $\mu(\theta_1)$ is the independent variable of the function $\hat{v}$, the sender's expected utility function is a piecewise linear function that divides the domain $[0,1]$ into $n$ disjoint \emph{intervals} $I_1, I_2,\cdots,I_n$. Each interval $I_i$ corresponds to a action $a^*(I_i)$ that the receiver's best response under $\mu$ is $a^*(I_i)$ for any posteriors $\mu$ that $\mu(\theta_1)\in I_i$. In particular, let the interval corresponding to the sender's favorite action $a_1$ be $I_{i^*}$. Therefore, the sender's expected utility function $\hat{v}$ corresponds to a partition $\hat{P}=(\{I_i\}_{i=1}^n,\{\alpha\}_{i=0}^n)$ of $[0,1]$. When $\mu(\theta_1)\in(\alpha_{i-1},\alpha_i)$ for all $1\le i\le n$, the best response for the receiver is $a^*(I_i)$. When $\mu(\theta_1)=\alpha_i$ for some $i$, the receiver breaks the tie in favor of the sender. We then present the definition of \emph{quasiconcave function} and \emph{quasiconcave closure}.

\begin{definition}[quasiconcave function]
A function $f:[0,1]\rightarrow\mathbb{R}$ defined on [0,1] is quasiconcave if for all $x,y\in [0,1]$ and $\lambda\in[0,1]$, we have $f(\lambda x+ (1-\lambda)y)\ge\min\{f(x),f(y)\}$.
\end{definition}

\begin{definition}[quasiconcave closure]
The quasiconcave closure of a function $f:[0,1]\rightarrow\mathbb{R}$ defined on $[0,1]$ is the pointwise lowest quasiconcave and upper semicontinuous function that is no less than $f$ at every point on $[0,1]$. \shuran{What does majorize mean? It seems to be an uncommon math jargon. Maybe give an simpler explanation before the formal definition, such as the quasiconcave closure is the minimal quasiconcave function that is no less than f at every point.}
\end{definition}
{\jiahao{added a figure for quasiconcave closure}}
\begin{figure}[!h]
     \centering
     \begin{subfigure}[b]{0.33\textwidth}
         \centering
         \includegraphics[width=\textwidth]{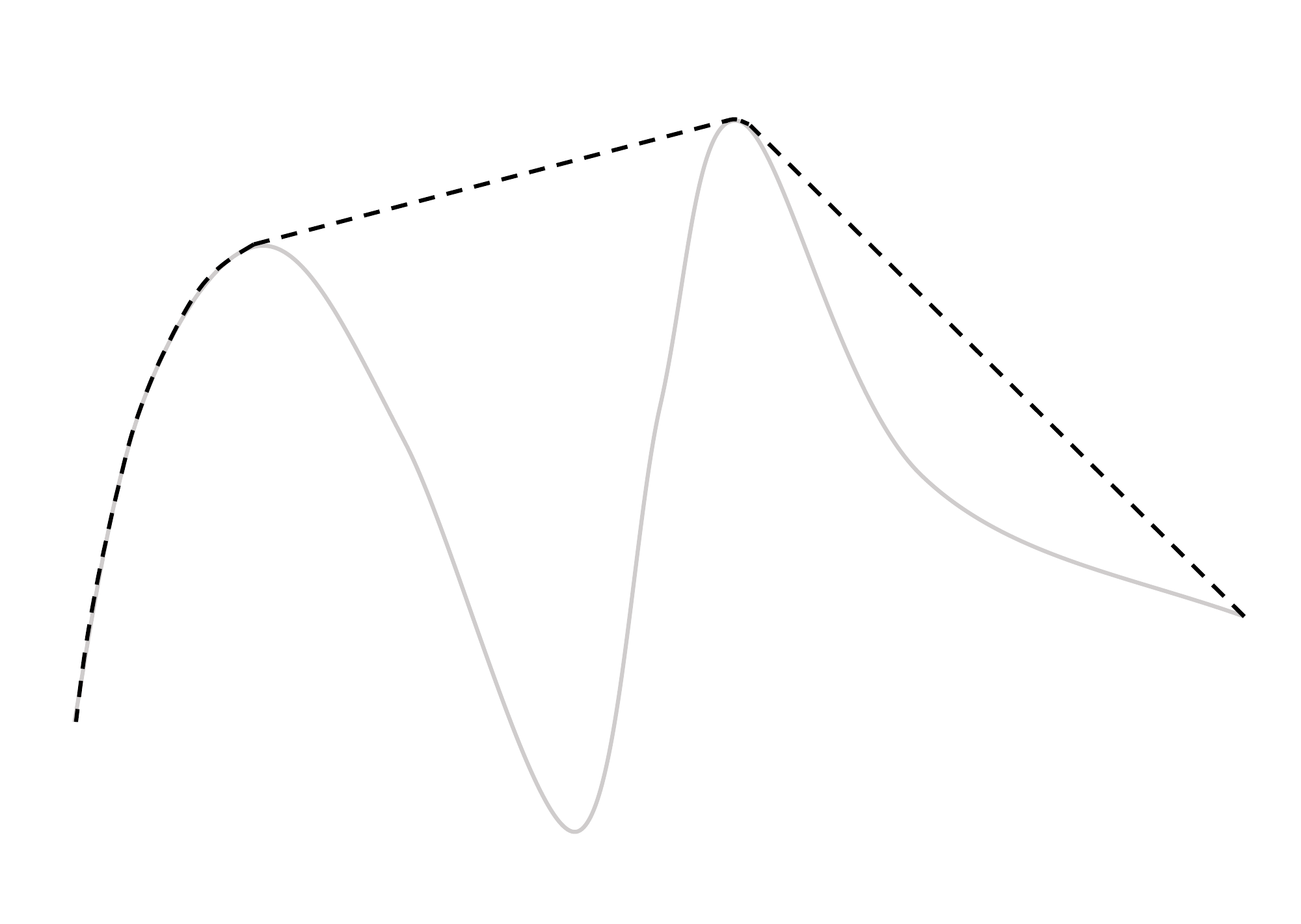}
     \end{subfigure}
     \begin{subfigure}[b]{0.33\textwidth}
         \centering
         \includegraphics[width=\textwidth]{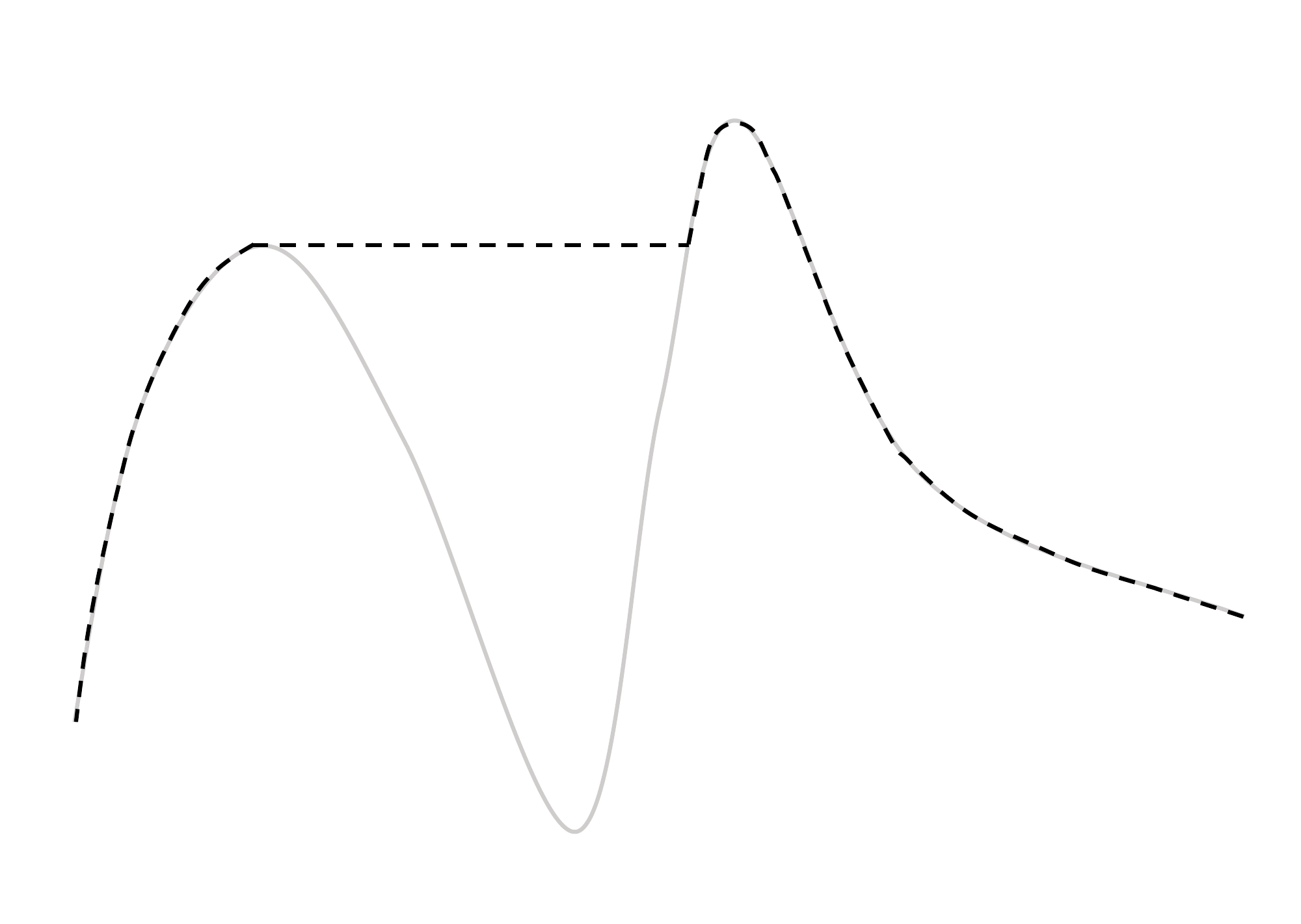}
     \end{subfigure}
     \caption{A function with its concave (left) and quasiconcave (right) closures.}
     \label{fig: quasi-example}
\end{figure}

Next, we present a useful lemma of the quasiconcave closure of the sender's expected utility function. It says the quasiconcave closure of $\hat{v}$ is also a piecewise function and shows the monotonicity of the quasiconcave closure. Above we say the function $\hat{v}$ corresponds to the partition $\hat{P}$. It is actually true that any piecewise function $f$ on $[0,1]$ corresponds to a partition $P$ of [0,1] consisting of all subdomains of $f$. The proof is simple and omitted. Now we use this fact to formally define the monotonicity of the piecewise function with respect to intervals.

{\jiahao{formal math for the following definitions, lemma and theorem}}
\begin{definition}[monotone with respect to intervals]\label{def: mono_piece}
    Consider a one-dimensional piecewise function $f:[0,1]\rightarrow\mathbb{R}$ corresponding to a partition $\{I_i\}_{i=1}^n$ of $[0,1]$. $f$ is increasing (decreasing) w.r.t intervals if for all $i,j\in[n]$ that $i<j$ and for all $x,y$ that $x\in I_i,y\in I_j$, $
        f(x)\le(\ge) f(y)$.
\end{definition}

Then we formally present our lemma of the quasiconcave closure as follows.
\begin{lemma}\label{lem:quasiconcave}
      Let $\overline{V}$ be the quasiconcave closure of the sender's utility function $\hat{v}$. Then we have
      \begin{enumerate}
          \item The quasiconcave closure $\overline{V}$ is a piecewise function corresponding to a partition $\overline{P}=(\{J_j\}_{j=1}^m,\{\beta_j\}_{j=0}^m)$ of $[0,1]$ such that $\overline{V}$ is continuous on each interval $J_j$ for all $j\in[m]$ and $\beta_j$ is a discontinuity for $1\le j<m$.
          \item There exists $i^*\in[n]$ and $j^*\in[m]$ such that $I_{i^*}\subset J_{j^*}$, i.e. there exists an interval in $\overline{P}$ such that  it contains the interval $I_{i^*}$ corresponding the sender's favorite action $a_1$.
          \item For all $j<j^*$, $\overline{V}$ is increasing w.r.t intervals and for all $j>j^*$, $\overline{V}$ is decreasing w.r.t intervals.
      \end{enumerate}
\end{lemma}

Now we are ready to describe our main result in this section. It shows how to construct a smoothed quasiconcave closure directly from the quasiconcave closure and there is no gap between the optimal signaling scheme and optimal ex-post IR signaling scheme if and only if such a smoothed quasiconcave closure is concave. 
\shuran{Move the definition out of the theorem and formally define the smoothed quasiconcave closure. Add pictures to illustrate the definitions.}

\begin{definition}[smoothed quasiconcave closure]\label{def: SQC}
Let the quasiconcave closure $\overline{V}$ corresponds to the partition $\overline{P}=(\{J_j\}_{j=1}^m,\{\beta_j\}_{j=0}^m)$ of $[0,1]$ mentioned in \cref{lem:quasiconcave}. The smoothed quasiconcave closure $\Gamma$ is a function defined on [0,1] that for all $1\le j\le m$, if $x\in[\beta_{j-1},\beta_j]$,
    \[
        \Gamma(x)=\frac{\overline{V}(\beta_j)-\overline{V}(\beta_{j-1})}{\beta_j-\beta_{j-1}}(x-\beta_{j-1})+\overline{V}(\beta_{j-1}).
    \]                                                                     
\end{definition}
That is, we smoothen the quasiconcave closure by connecting adjacent endpoints $(\beta_{j-1},\overline{V}(\beta_{j-1}))$ and $(\beta_j,\overline{V}(\beta_j))$ for all $1\le j\le m$. We provide a simple example in \cref{fig: SQC_example} to illustrate the definition of smoothed quasiconcave closure.

\begin{figure}[!h]
     \centering
     \begin{subfigure}[b]{0.37\textwidth}
         \centering
         \includegraphics[width=\textwidth]{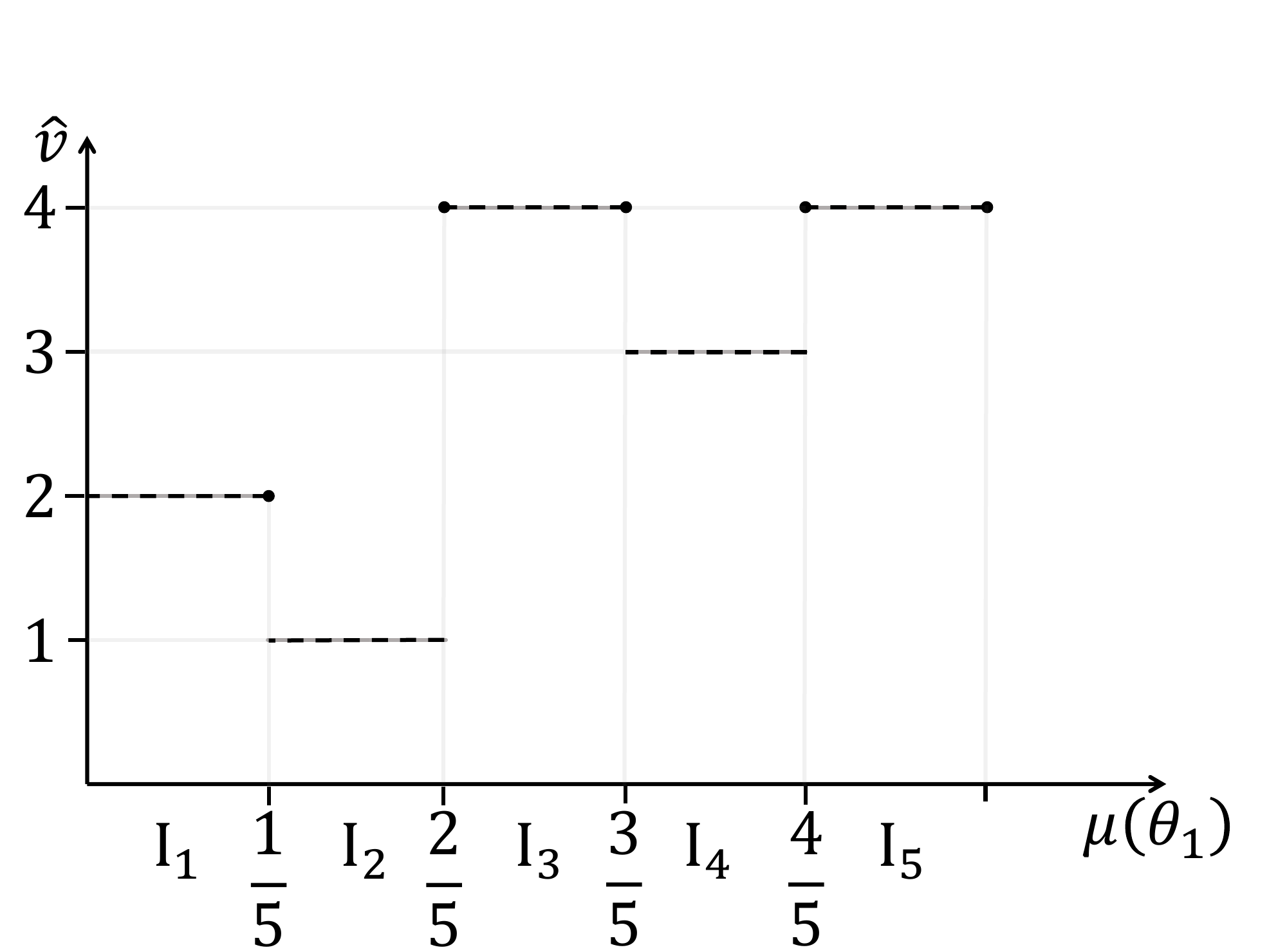}
         \caption{sender's utility function}
         \label{fig: smoothed-concave}
     \end{subfigure}
     \begin{subfigure}[b]{0.37\textwidth}
         \centering
         \includegraphics[width=\textwidth]{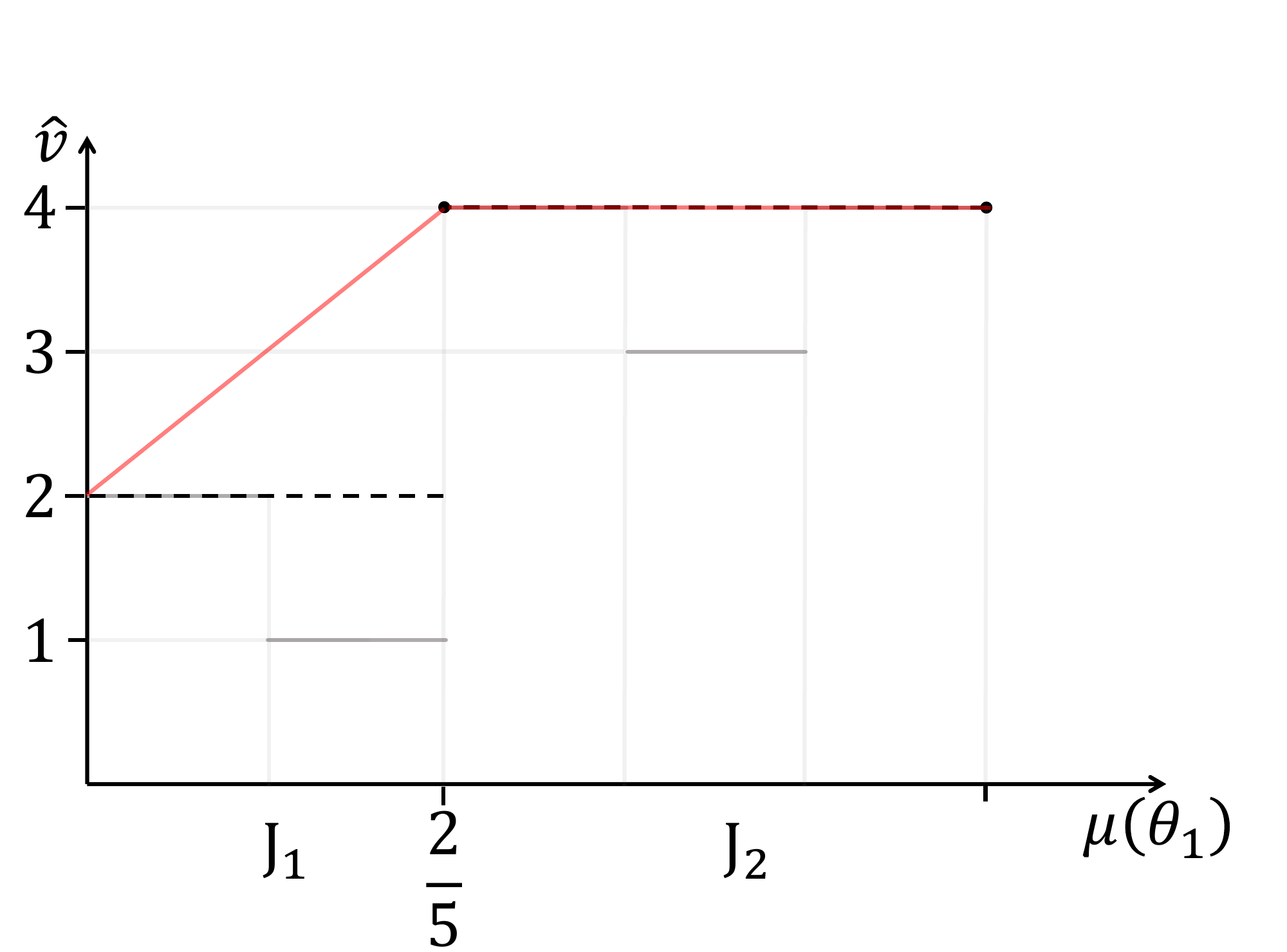}
         \caption{(smoothed) quasiconcave closure}
         \label{fig: smoothed-quasiconcave}
     \end{subfigure}

        \caption{A simple example of smoothed quasiconcave closure. Left: the curve of the sender's expected utility function $\hat{v}$ corresponding to a partition $\hat{P}=(\{I_i\}_{i=1}^5,\{\alpha_i\}_{i=0}^5)$ of $[0,1]$ that $\alpha_i=\frac{i}{5}$ for all $0\le i\le5$. Right: the dashed line is the quasiconcave closure $\overline{V}$ corresponding to a partition $\overline{P}=(\{J_1,J_2\},\{\beta_0=0,\beta_1=\frac{2}{5},\beta_2=1\})$ and the red line is the smoothed quasiconcave closure.}
        \label{fig: SQC_example}
\end{figure}

\begin{theorem}\label{thm:n*2}
The optimal signaling scheme is ex-post IR if and only if $\Gamma$ is concave.
\end{theorem}

\begin{remark}
    Given $\InAbs{A}=n$, \cref{thm:n*2} provides an algorithm to determine if the optimal signaling scheme is ex-post IR with the $O(n\log n)$ complexity, which is better than simply running a linear program to solve the optimal signaling scheme and check whether it is ex-post IR. In detail, we can sort each row of the receiver's utility matrix according to the first entry of each row $u(a_i,\theta_1)$ for all $i\in[n]$ with complexity $O(n \log(n))$. Then we can solve the sender's expected utility function $\hat{v}$ with complexity $O(n)$. Because when the row of the receiver's utility matrix is sorted, we can compute the threshold of each piece of $\hat{v}$ by inequality $\mu(\theta_1)u(a_i,\theta_1)+\mu(\theta_2)u(a_i,\theta_2)\le\mu(\theta_1)u(a_{i+1},\theta_1)+\mu(\theta_2)u(a_{i+1},\theta_2)$ for all $1\le i<n$. At last, we can construct $\Gamma$ and check whether it is concave in $O(n)$ complexity according to \cref{thm:n*2}.
\end{remark}

\shuran{This paragraph is difficult to understand. Could you explain it in a more intuitive way? Using figures might help. Otherwise, it is not necessary to have this paragraph.} 

The intuition behind \cref{thm:n*2} lies in the two fundamental properties of the quasiconcave closure $\overline{V}$: (1) $\overline{V}(\mu_2)\ge\min\{\overline{V}(\mu_1),\overline{V}(\mu_3)\}$ for all $\mu_1\le\mu_2\le\mu_3\in[0,1]$ (quasiconcavity); (2) $\overline{V}$ is increasing (decreasing) w.r.t intervals on the left (right) of Interval $J_{j^*}$ (monotonicity). On one hand, if $\Gamma$ is concave, it is the concave closure of the expected utility function. Then for any prior, the optimal signaling scheme will induce better actions for the sender than the best response under prior because of the monotonicity, which means the optimal signaling scheme is ex-post IR. On the other hand, if the optimal signaling scheme is ex-post IR, we can prove that the concave closure is actually the smoothed quasiconcave closure by the definition of quasiconcavity. Hence the smoothed quasiconcave closure is concave.

{\jiahao{changed examples and figures.}}

Finally, we provide a toy example for that can help us better interpret the result of \cref{thm:n*2} through visualization.
In this example, a fixed receiver plays Bayesian persuasion separately with two different senders. We show that the smoothed quasiconcave closure of the first sender is not concave, and thus, the optimal signaling scheme is not ex-post IR under some prior. In contrast, the smoothed quasiconcave closure of the second sender is concave, and thus, the optimal signaling scheme is ex-post IR for any prior. Let the receiver's utility matrix and two senders' utility matrices be as follows.
\begin{table}[!ht]
\centering
 \scalebox{0.95}{\begin{tabular}{|c|c|c|}
     \hline $v$ & $\theta_1 = 0$ & $\theta_2 =1$ \\
     \hline $a_1$ & $4$ & $4$ \\
     \hline $a_2$ & $3$ & $3$ \\
     \hline $a_3$ & $2$  & $2$ \\
     \hline $a_4$ & $1$  & $1$ \\
     \hline
\end{tabular} 
~\qquad
 \begin{tabular}{|c|c|c|}
     \hline $v$ & $\theta_1 = 0$ & $\theta_2 =1$ \\
     \hline $a_1$ & $4$ & $4$ \\
     \hline $a_2$ & $7/2$ & $7/2$ \\
     \hline $a_3$ & $2$  & $2$ \\
     \hline $a_4$ & $1$  & $1$ \\
     \hline
\end{tabular} 
~\qquad
 \begin{tabular}{|c|c|c|}
     \hline $v$ & $\theta_1 = 0$ & $\theta_2 =1$ \\
     \hline $a_1$ & $8$ & $0$ \\
     \hline $a_2$ & $7$ & $3$ \\
     \hline $a_3$ & $0$  & $8$ \\
     \hline $a_4$ & $3$  & $7$ \\
     \hline
\end{tabular}}
\caption{From left to right: the utility of the first sender, the utility of the second sender, the utility of the receiver.}
\label{tab: quasi}
\end{table}

The utility functions of these senders are both state independent. The first sender's utility is $4,3,2,1$ when the receiver's best response is $a_1,a_2,a_3,a_4$ respectively. The second sender's utility is $4,7/2,2,1$ when the receiver's best response is $a_1,a_2,a_3,a_4$ respectively. The receiver takes action $a_3$ when the posterior probability of state being zero $\mu_s(0)\le\frac{1}{4}$, takes action $a_4$ when $\frac{1}{4}\le\mu_s(0)\le\frac{1}{2}$, takes action $a_2$ when $\frac{1}{2}\le\mu_s(0)\le\frac{3}{4}$, and takes action $a_1$ otherwise. The figures of the senders' expected utility functions are \cref{fig: NIR-utility,fig: IR-utility}. The optimal signaling schemes can be visualized from \cref{fig: NIR-concave,fig: IR-concave}. In the first case, when the prior is $\frac{1}{2}\le\mu_s(0)\le\frac{3}{4}$, the optimal disclosure policy is divided the prior into two posterior $s_1,s_2$, $\mu_{s_1}(0)=0$ and $\mu_{s_2}(0)=\frac{3}{4}$ with the best response being $a_3$ and $a_1$ respectively. But the best response under the prior is $a_2>a_3$. So the sender is worse off under $s_1$, which means the optimal signaling scheme is not ex-post IR. In contrast, the sender is never worse off under any signal of the optimal disclosure policy. Meanwhile, the figures of (smoothed) quasiconcave closure $\Gamma$ are in \cref{fig: NIR-quasiconcave,fig: IR-quasiconcave}. The difference is that the smoothed quasiconcave closure is not concave in the first case while is concave in the second case.

\begin{figure}[!h]
     \centering
     \begin{subfigure}[b]{0.32\textwidth}
         \centering
         \includegraphics[width=\textwidth]{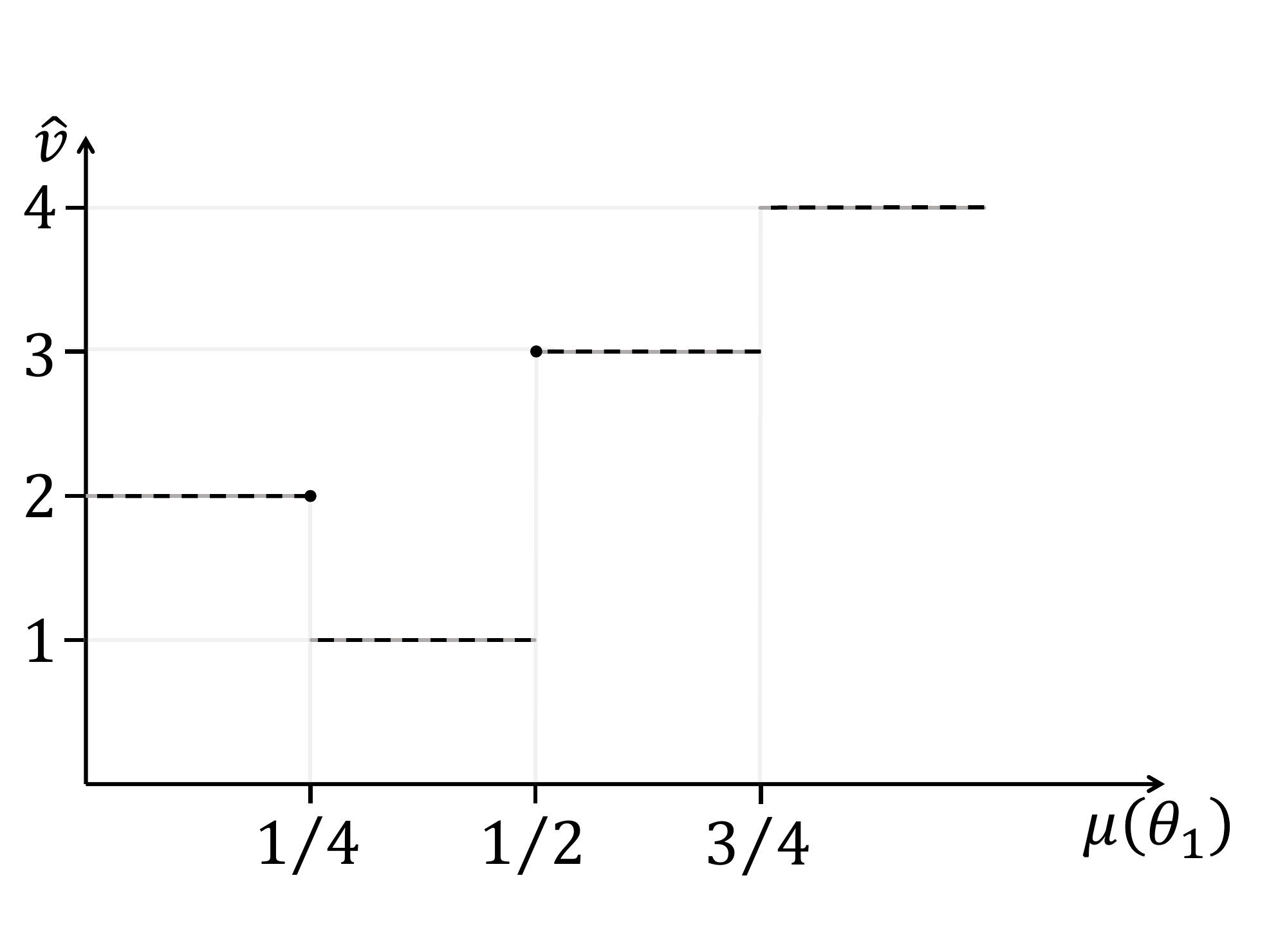}
         \caption{sender's utility function}
         \label{fig: NIR-utility}
     \end{subfigure}
     \begin{subfigure}[b]{0.32\textwidth}
         \centering
         \includegraphics[width=\textwidth]{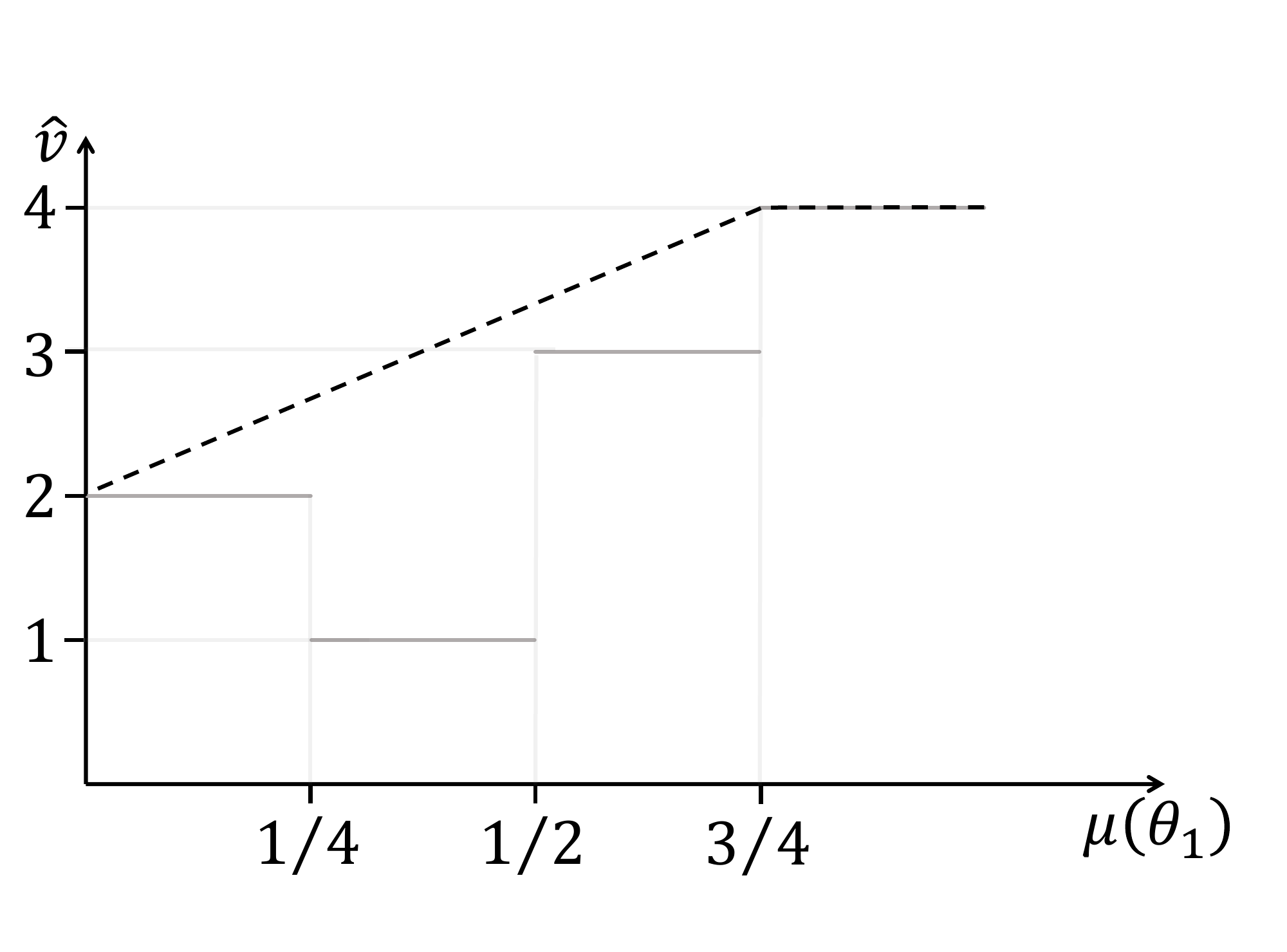}
         \caption{concave closure}
         \label{fig: NIR-concave}
     \end{subfigure}
     \begin{subfigure}[b]{0.32\textwidth}
         \centering
         \includegraphics[width=\textwidth]{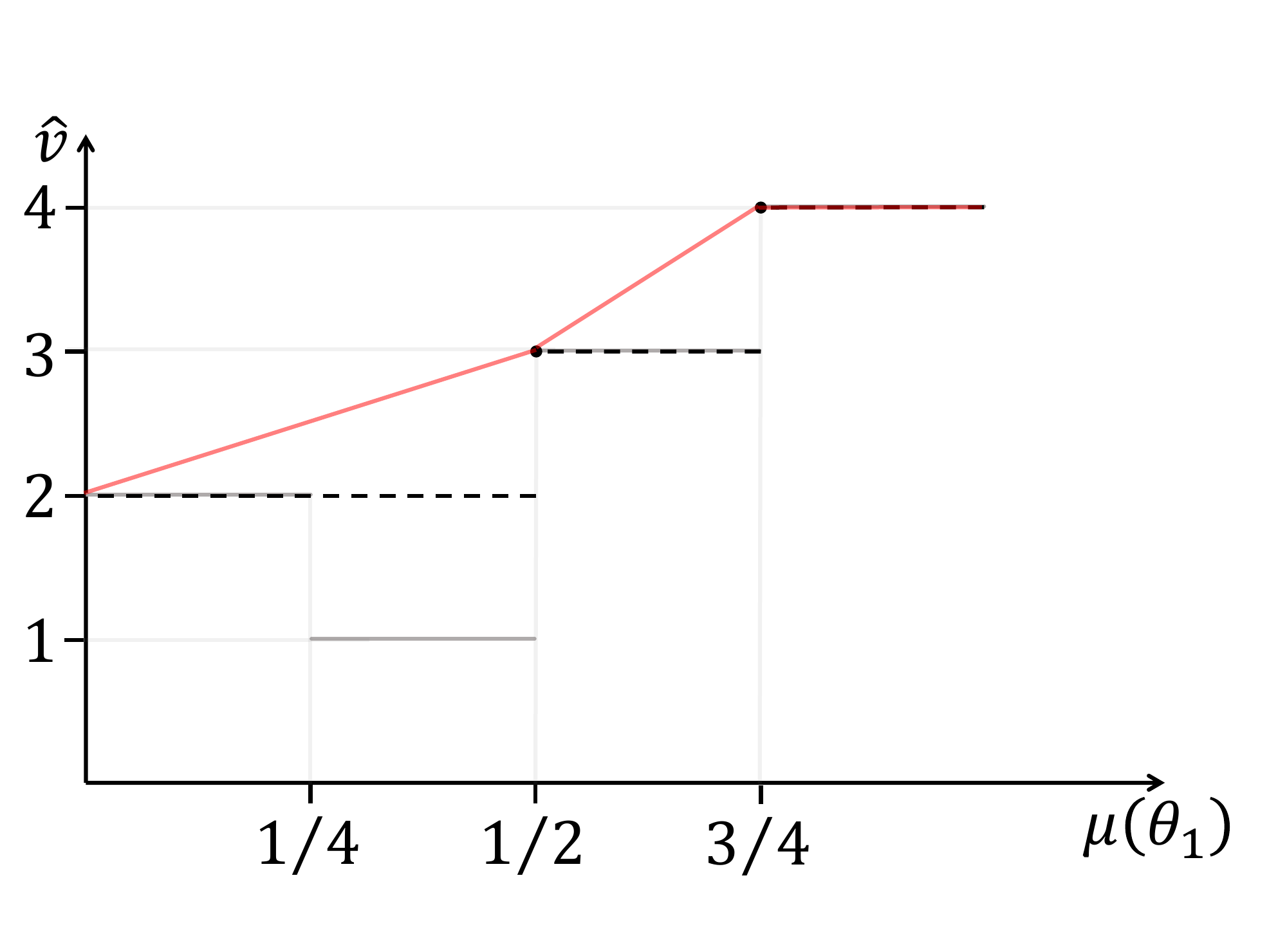}
         \caption{(smoothed) quasiconcave}
         \label{fig: NIR-quasiconcave}
     \end{subfigure}
        \caption{The example corresponding to the first sender (\cref{tab: quasi}) that the optimal signaling scheme is not ex-post IR. From left to right: the curve of the first sender's utility function (dashed line in \cref{fig: NIR-utility}), the concave closure of the first sender's utility function (dashed line in \cref{fig: NIR-concave}), the quasiconcave closure (dashed line in \cref{fig: NIR-quasiconcave}) and the smoothed quasiconcave closure (red line in \cref{fig: NIR-quasiconcave}) of the first sender's expected utility function. The smoothed quasiconcave closure is not concave. \shuran{Do we need to plot the concave closure here? Would it be better to plot the quasiconcave closure?}} 
        \label{fig: NIR}
\end{figure}

\begin{figure}[!h]
     \centering
     \begin{subfigure}[b]{0.32\textwidth}
         \centering
         \includegraphics[width=\textwidth]{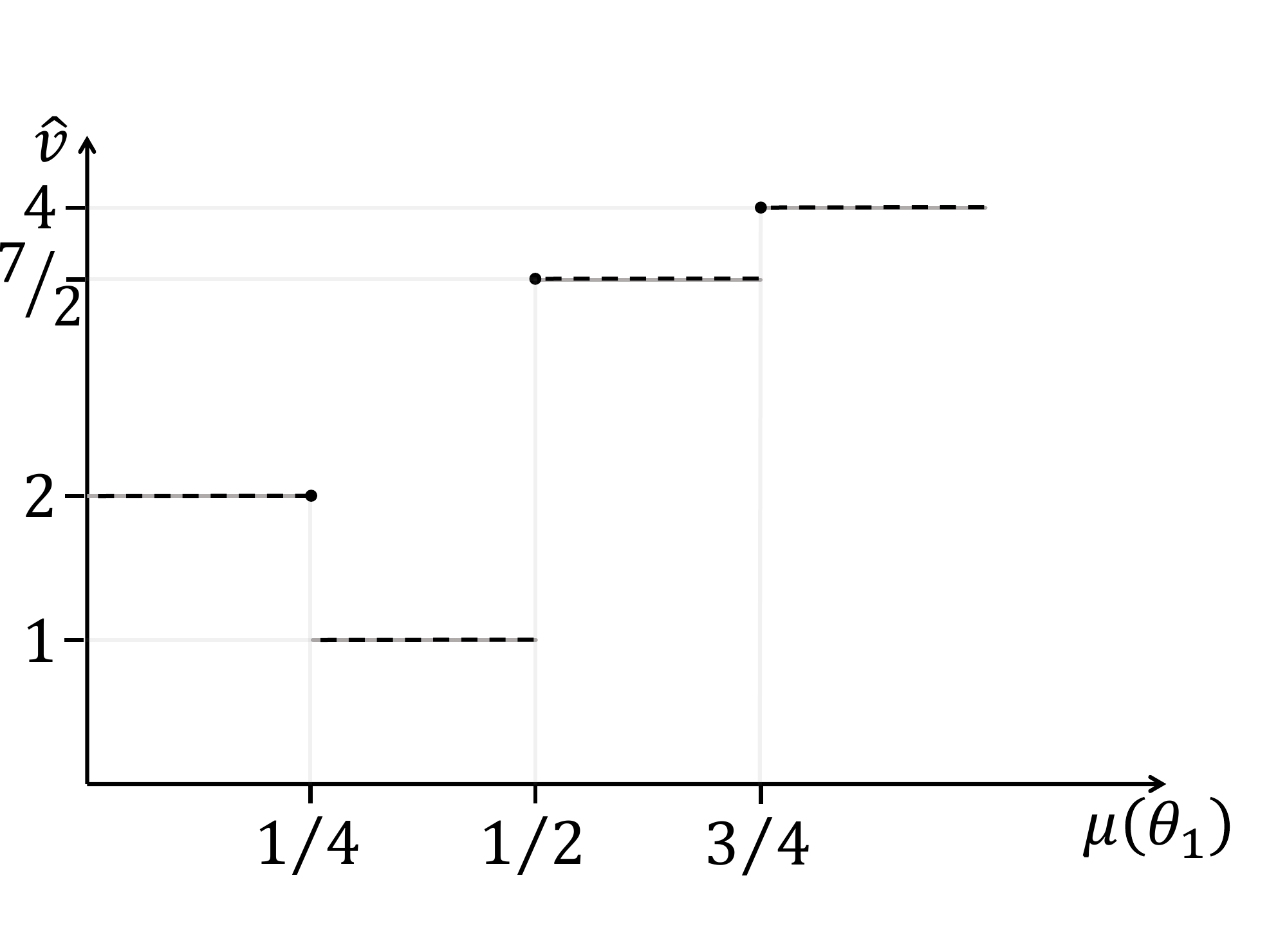}
         \caption{sender's utility function}
         \label{fig: IR-utility}
     \end{subfigure}
     \begin{subfigure}[b]{0.32\textwidth}
         \centering
         \includegraphics[width=\textwidth]{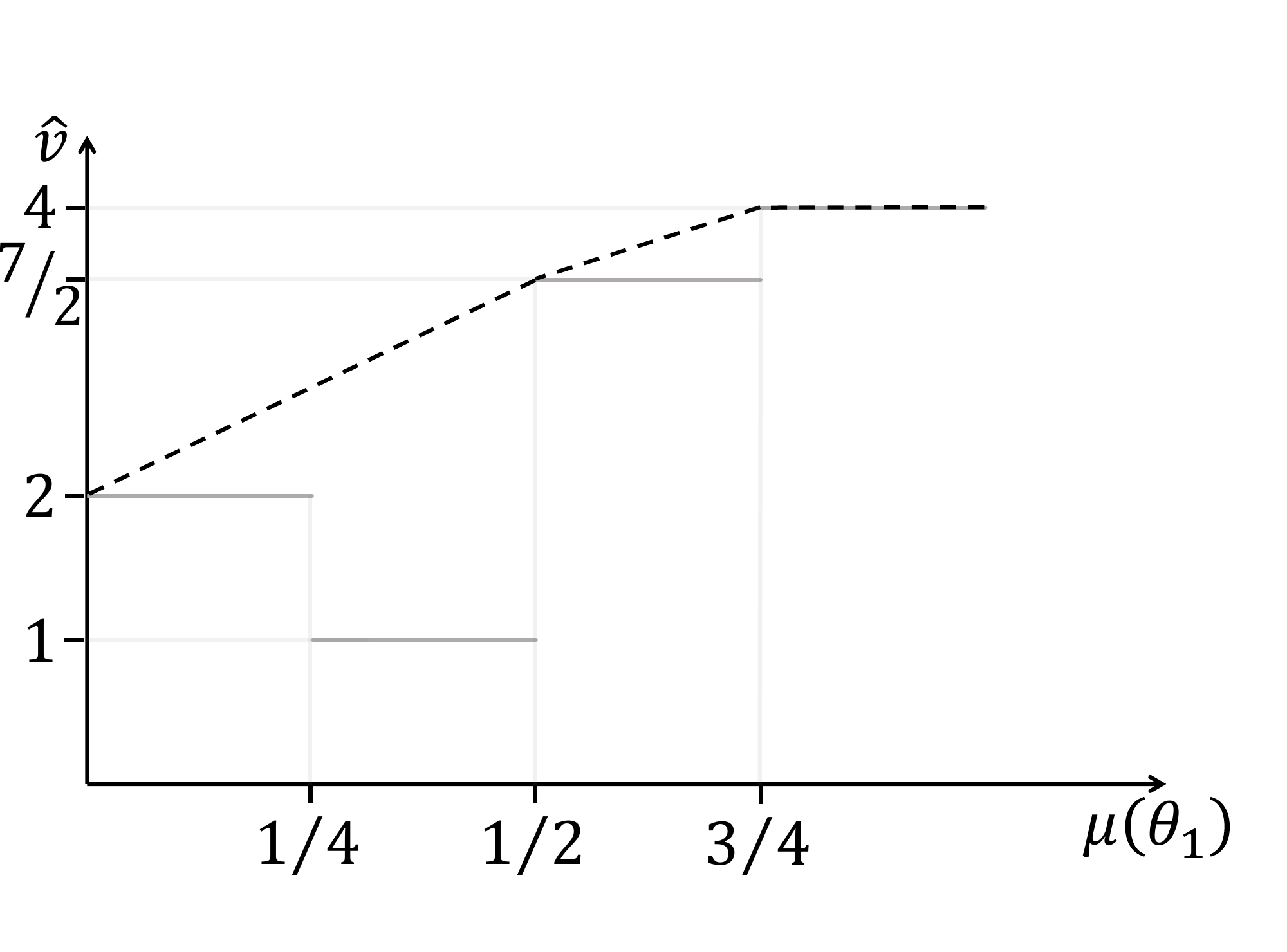}
         \caption{concave closure}
         \label{fig: IR-concave}
     \end{subfigure}
     \begin{subfigure}[b]{0.32\textwidth}
         \centering
         \includegraphics[width=\textwidth]{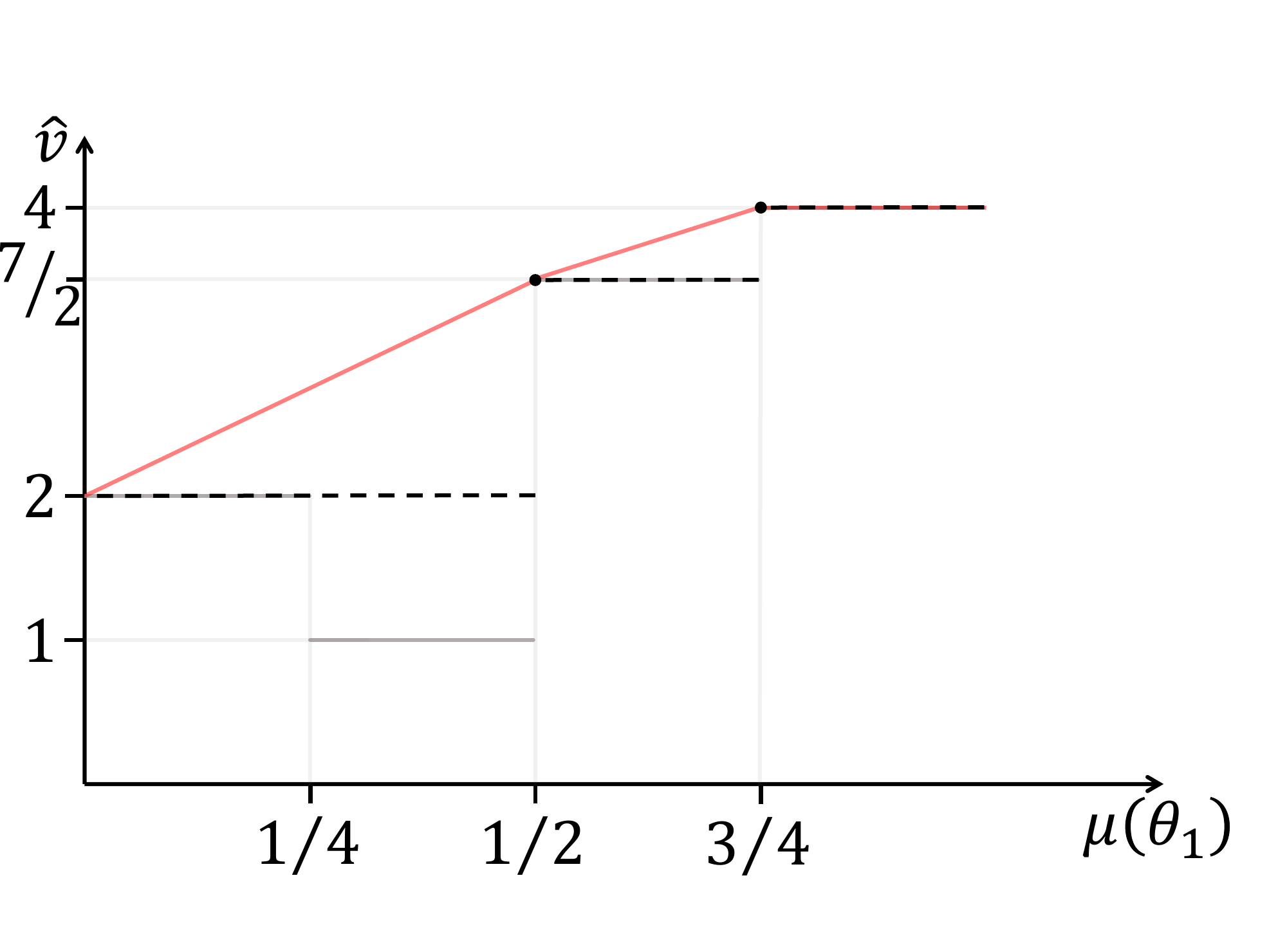}
         \caption{(smoothed) quasiconcave}
         \label{fig: IR-quasiconcave}
     \end{subfigure}
        \caption{The example corresponding to the second sender (\cref{tab: quasi}) that the optimal signaling scheme is ex-post IR. From left to right: the curve of the second sender's expected utility function (dashed line in \cref{fig: IR-utility}), the concave closure of the second sender's utility function (dashed line in \cref{fig: IR-concave}), the quasiconcave closure (dashed line in \cref{fig: NIR-quasiconcave}) and the smoothed quasiconcave closure (red line in \cref{fig: NIR-quasiconcave}) of the second sender's utility function. The smoothed quasiconcave closure is concave.} 
        \label{fig: IR}
\end{figure}

\subsection{No gap in the bilateral trade market}\label{subsection:bilateral}
{\color{black} In this section, we introduce a class of Bayesian persuasion problems in which the sender will not be affected by the ex-post IR constraint, which we term as \emph{trading games}. The discovery of trading games is inspired by the bilateral trade persuasion problem studied in \cite{bergemann2015limits,glode2018voluntary}. We first observe that the bilateral trade persuasion problem always has an optimal signaling scheme that is ex-post IR. Drawing from this insight, we identify three key properties that guarantee the ex-post IR of the optimal signaling scheme. These three properties yield a broader class of trading games that preserve the ex-post individual rationality of the optimal signaling scheme.}  We first introduce the bilateral trade pricing problem studied in \cite{bergemann2015limits,glode2018voluntary} and show that there always exists an optimal signaling scheme that is ex-post IR in this problem.

\paragraph{Bilateral trade.} A monopolist sells an item to a consumer, who has $n$ possible values $\theta\in\Theta$: $0<\theta_1<\theta_2<...<\theta_n$. The monopolist sets a price $a\in A$ and the trade only happens when the consumer's value is higher than the price. The monopolist will not set a price between any two closest values $\theta_i$ and $\theta_{i+1}$ for some $i$ because the monopolist can increase the price to $\theta_{i+1}$ without decreasing the probability of making a deal. Therefore, we can let $A=\Theta$ without loss of generality. The monopolist doesn't know the consumer's value but shares a prior $\mu_0$ of the value with the consumer. Before the trade, the consumer commits to a signaling scheme $\phi$. \shuran{It might be good to mathematically define the signaling scheme and the posteriors.} Then the monopolist sets the price according to the posteriors $\mu_s$ for any realized signal $s\in S$. The utility of the consumer is  $v(a,\theta)=(\theta-a)\cdot\mathbbm{1}_{\theta\ge a}$ and the utility of the monopolist is $u(a,\theta)=a\cdot\mathbbm{1}_{\theta\ge a}$.

\begin{table}[!h]
    \centering
\scalebox{0.95}{\begin{tabular}{|c|c|c|c|c|}
     \hline $v$ & $\theta_1$ & $\theta_2$ &   $ \cdots$ & $\theta_n$   \\
     \hline $a_1(=\theta_1) $ & $\theta_1$  & $ \theta_1$ &  $ \cdots$ & $ \theta_1$  \\
     \hline $a_2$  & $0$ & $\theta_2$ &  $ \cdots$ & $ \theta_2$ \\
     \hline $ \cdots$   & $ \cdots$  & $ \cdots$  &  $\cdots$ & $\theta_{n-1}$  \\ 
     \hline  $a_n$  & $0$ & $0$ &  $ \cdots$  & $\theta_n$  \\
     \hline
\end{tabular}}
    \caption{The utility of receiver in the bilateral trade}
    \label{bilateral_recv}
\end{table}

\begin{proposition}\label{bilateral_nogap}
     In the bilateral trade persuasion problem defined in~\cite{bergemann2015limits}, {\color{black} there exists an optimal signaling scheme that is ex-post IR.}
\end{proposition}

We then extend the bilateral trade game to a general class of games called  \emph{trading games}, in which the ex-post IR constraint does not affect the sender's expected utility. 
We first make a few observations about the bilateral trade problem: (1) the utility of the receiver is upper-triangular (2) the non-zero part (see \cref{bilateral_recv}) corresponding columns of the receiver utility is increasing from top to bottom (3) the sum of the non-zero part corresponding columns of the sender utility and the receiver utility is a constant (see formal definition in \cref{def:trading}). In the context of bilateral trade games, the first feature means that the buyer won't buy the item if the price is higher than its value; the second feature means that when the price is higher, the buyer's utility is lower; the third feature means that once the trade happens, the social welfare is a constant equals to the buyer's value. We extend the bilateral trade game to a general class of games called \emph{trading games}. As long as a game satisfies the aforementioned properties, we call it a trading game (see \cref{def:trading}). 

\begin{definition}\label{def:trading}(Trading game)
    A game $(A,\Theta,v, u)$ is a trading game if $\InAbs{A}=\InAbs{\Theta}=n$ for any integer $n$ and the following conditions hold
    \begin{enumerate}
        \item u is upper-triangular and non-negative: $u(a_i,\theta_k)\ge0$ for all $i,k\in[n]$ and $u(a_i,\theta_k)=0$ if $i>k$.
        \item the non-zero part corresponding columns of the receiver utility is increasing: $u(a_i,\theta_k)\le u(a_j,\theta_k)$ for all $i,j,k\in[n]$ that $i<j\le k$.
        \item $v(a_i,\theta_k)+u(a_i,\theta_k)=c_k$ where $c_k$ is a positive constant for all $i,k\in[n]$ that $i\le k$.
    \end{enumerate}
\end{definition}

These three conditions generally hold in trading problems with a buyer and a seller that do not have to be a bilateral trade problem. The following example is about the Bayesian persuasion problem in the first price auction (FPA) with a reserve price, which we show is also a trading game.
 
 \shuran{It might be better to define the auction problem here and show that there exist trading games that are not bilateral trade games.} 

 \begin{example}[First-price auction]
Consider the first-price auction with multiple bidders, when we set reserve prices, we only care about the maximum bid. It is useful to abstract bidders in the auction as a single "mega-bidder" whose value is the maximum of values and who always bids the maximum of all bids (see \cite{feng2020reserve}). Now we can consider the pricing problem between the seller(receiver) and the representative winning bidder(sender). The state is the value of the bidder $\theta_1\le\theta_2\le...\le\theta_n$. The action is the reserve price set by the seller. Notice that the seller will never set a reserve price between some $\theta_{i-1}$ and $\theta_i$ because he has the incentive to increase the reserve price to exact $\theta_i$. Then the receiver has exactly n actions $a_1,...,a_n$ which $a_i$ is setting the reserve price at $\theta_i$. The bidder's bidding function $b(r,v)$, which is determined by the reserve price $r$ and the bidder's value $v$, is increasing w.r.t to the reserve price $r$. We have $v(a_i,\theta_j)=\theta_j-b(\theta_i,\theta_j)$ if $i\le j$ otherwise $v(a_i,\theta_j)=0$ and $u(a_i,\theta_j)=b(\theta_i,\theta_j)$ if $i\le j$ otherwise $u(a_i,\theta_j)=0$. 


\end{example}

 We then show that as long as the game is a trading game, there exists an optimal signaling scheme that is ex-post IR.

\begin{theorem}\label{extendBBM}
    For all trading games, there exists an optimal signaling scheme that is ex-post IR.
\end{theorem}

\jiahao{added how the properties are used.} To illustrate the idea of the proof, we first provide a lemma about the property of the receiver's best response, which relies on the first two properties of \cref{def:trading}.

\begin{lemma}\label{lem:upper}
Assume the utility matrix $u$ is upper-triangular, non-negative, and the non-zero part corresponding columns of $\mu$ is increasing, for any subset $S\subset[n]$, there exists a posterior belief $\mu_S$ supported on $S$ that the receiver is indifferent with actions $a_k$ for all $k\in S$ under the posterior $\mu_S$.
\end{lemma}

\shuran{This paragraph is difficult to read. Could you describe the process more accurately using math? Otherwise, it is not necessary to have it here.} 


Based on how the receiver best responds, we show that there is a specific structure of ex-post IR signaling schemes. Recall that a signaling scheme induces a set of possible posteriors. We now describe how to construct all posteriors for such signaling schemes step by step. In the first step $t=1$, let $\mu_1$ be the prior $\mu_0$ and $S_1=\{1,2,\cdots,n\}$. We compute a posterior $\mu_{S_1}$ such that the receiver is in indifferent among $\theta_i$ for all $i\in S_1$. Then we subtract certain proportions of $\mu_{S_1}$ from $\mu_1$ until some supports of $\mu_0$ become zero. In detail, let $\mu_2=\mu_1-p_1\mu_{S_1}$ such that $p_1=\min_{k\in S_1}\frac{\mu_1(\theta_k)}{\mu_{S_1}(\theta_k)}$.
$\mu_{S_1}$ is the first posterior we want and its probability in the signaling scheme is $p_1$. Let $S_2$ be the indexes of the support of $\mu_2$, i.e., $S_2=\{i|\mu_2(\theta_i)>0\}$. Given $\mu_2$ and $S_2$, we repeat this process at $t=2$ and so forth. This procedure will stop in finite steps because at least one component $\mu_t$ will become zero for $t\ge 1$. Suppose it stops in $m$ steps. The signaling scheme consist of $m$ posteriors $\mu_{S_1},\cdots,\mu_{S_m}$ and $\mu_0=\sum_{t=1}^mp_t\mu_{S_t}$. We know that these posteriors exist because of \cref{lem:upper}. We will show that $S_m=\{k^*\}$ which means that $\theta_{k^*}$ is the last component that becomes zero. Therefore, we end up dividing the prior into several posterior beliefs that $a_{k^*}$ is one of the best responses under each of them. The sender hence always enjoys a utility better than the utility under the prior, which means the signaling scheme is ex-post IR. We then complete the proof by using the third property of \cref{def:trading} to show that such a signaling scheme is optimal. The detailed proof is in \cref{proof:extendBBM}.



\section{Comparison of the sender's optimal utility}\label{section:compare}
In this section, we compare the sender's expected utility under different credibility requirements. We consider four information transmission models with different credibility levels: Bayesian persuasion, ex-post IR Bayesian persuasion, 
 credible persuasion~\cite{lin2022credible}, and cheap talk~\cite{crawford1982strategic}. In many previous works about Bayesian persuasion with partial commitment, such as \cite{min2021bayesian,lin2022credible}, the optimal sender utility achieved by their model is between the Bayesian persuasion and cheap talk. Therefore, it is natural to compare the sender's optimal utility of our ex-post IR Bayesian persuasion with the commitment power of Bayesian persuasion and cheap talk. Besides, since our model considers the credibility of the sender, we compare the ex-post IR Bayesian persuasion to another model called \emph{credible persuasion} in the paper by \cite{lin2022credible}, which also considers the credibility of the sender. Given the sender's utility $v$, receiver's utility $u$, the optimal sender utilities gained in these four models is the functions of prior $\mu_0$ and denoted as $V_{\textsc{BP}}(\mu_0)$, $V_{\textsc{expost}}(\mu_0)$, $V_{\textsc{credible}}(\mu_0)$, $V_{\textsc{cheap}}(\mu_0)$. We formally give the definition that the optimal sender utility achieved in one information transmission model is better than the optimal sender utility achieved in another information transmission model.
\begin{definition}
Consider two information transmission model $X, Y$. If for any $v,u,\mu_0$, we have $V_X(\mu_0)\ge V_Y(\mu_0)$, we say the optimal sender utility achieved in $X$ is better than the optimal sender utility achieved in $Y$, which is denoted by $V_X\succeq V_Y$.
\end{definition}


Next, we compare the optimal sender utility of the aforementioned four information transmission models. We give a brief introduction of credible persuasion in \cref{sec:credible}.

\begin{theorem}\label{thm:compare}
The ranking of the optimal sender's utility is as follows:
\begin{enumerate}
    \item $V^{\mathrm{BP}}\succeq V^{\mathrm{credible}},V^{\mathrm{expostBP}}\succeq V^{\mathrm{cheaptalk}}$.
    \item  There exists $v,u$ such that for any $\mu_0$, $V^{\mathrm{credible}}(\mu_0)\ge V^{\mathrm{expostBP}}(\mu_0)$.
    \item  There exists $v,u$ such that for any $\mu_0$, $V^{\mathrm{expostBP}}(\mu_0)\ge V^{\mathrm{credible}}(\mu_0)$.
\end{enumerate} 
\end{theorem}
The first part of \cref{thm:compare} that ex-post IR Bayesian persuasion and credible persuasion are partial commitment models between Bayesian persuasion and cheap talk. The second part of \cref{thm:compare} that the commitment power of ex-post IR Bayesian persuasion and credible persuasion differs in different cases based on a series of lemma that show when the optimal sender utility of ex-post IR Bayesian persuasion or credible persuasion is the same as the optimal sender utility of Bayesian persuasion. Then we can construct examples when one model's optimal sender utility is equal to the optimal sender utility of Bayesian persuasion and the other model's optimal sender utility is strictly smaller than the optimal sender utility of Bayesian persuasion.

\section{Conclusion}
In this paper, we studied ex-post IR Bayesian persuasion, a novel Bayesian persuasion model with restricted commitment power. We show that the ex-post IR constraint is a linear constraint; hence, linear programming can still solve the optimal ex-post IR signaling scheme. Our main results give broad conditions when there is no utility gap for the sender between the optimal signaling scheme and the optimal ex-post IR signaling scheme. At last, we show that, in the sense of the sender's optimal utility, our model and credible persuasion is an information transmission model between Bayesian persuasion and cheap talk. We further show when the optimal sender utility of our model is higher than the optimal sender utility of credible persuasion and vice versa.

\bibliographystyle{abbrvnat}
\bibliography{ref}

\newpage
\appendix
\section{Related literature}
A substantial body of previous research has studied the credibility of the sender in Bayesian persuasion. They consider the notion of credibility in different ways, which are all different from ours. In the 
work by \citet{nguyen2021bayesian}, the sender sends a signal to the receiver at some cost that depends on the signal and the realized world state, which makes the sender more credible. \citet{lipnowski2022persuasion} consider the sender to communicate with the receiver through an intermediary which can be influenced by the sender with some probability.
\citet{lin2022credible} characterize the notion of credibility by restricting the sender to only commit to a marginal distribution of signals rather than the whole signaling scheme. There are more works about the limitation of the commitment power via repeated interactions (\cite{best2016honestly,mathevet2022reputation}), verifiable information (\cite{hart2017evidence,ben2019mechanisms}), information control (\cite{luo2018strategies}), and mediation (\cite{salamanca2021value}).

Our paper considers two applications of Bayesian persuasion: bilateral trade and credence goods market. The information transmission models about these two applications has also been extensively studied in previous works. For the bilateral trade, \citet{ali2020voluntary} consider the welfare of the sender (consumer) in both monopolistic and competitive markets, and \citet{glode2018voluntary} consider the voluntary disclosure in the bilateral trade rather than the Bayesian persuasion. For the credence goods market, \citet{fong2014role} compare the market outcome with and without commitment and their model is extended by \citet{hu2023information}.

Technically, at a very high level, our work relies on two classic tools, which have also been successfully used in previous works. The first is quasiconcave closure. For example, \citet{lipnowski2020cheap} show that in cheap talk, when the sender's utility is state-independent, the optimal sender utility can be characterized by the quasiconcave closure of the sender's expected utility function. The second is the greedy algorithm. For example, \citet{lucier2015greedy} study mechanisms that use greedy allocation rules and show that all such mechanisms obtain a constant fraction of the optimal welfare.

\section{More application: no gap in the credence goods market}\label{subsection:credence}
In this section, we prove when there is no gap in the sender's expected utility led by the ex-post IR constraint in the \emph{credence goods market}. 

In the credence goods market, the sender's utility gap is bounded by the gap between the optimal signaling scheme and an optimal greedy signaling scheme. We again identify two key properties of the credence goods market, which we term \emph{cyclically monotonicity} and \emph{weak logarithmic supermodularity}, based on which we extend the credence goods market to a broader class of games where the gap led by the ex-post IR constraint can be bounded in the same way.

First, we give a brief introduction to the credence goods market. A credence good is a product or service whose qualities or appropriateness are hard to evaluate for consumers even after consumption. Common examples of credence goods include medical services, automobile repairs, and financial advice. For instance, patients who receive expensive medical treatments may struggle to know their value post-recovery. Car owners may question the necessity of replacing expensive parts of cars during repairs. Additionally, investors with limited financial knowledge may remain unsure about whether a financial product is suitable for long-term gains even if the product has realized short-term returns. On the flip side, expert sellers frequently possess superior insights into the appropriateness of credence goods and can use this knowledge to their advantage.

\paragraph{The persuasion problem for credence goods market} The client has a problem with an uncertain type which is drawn from a set $\Theta=\{\theta_1,...,\theta_n\}$ according to the common prior $\mu_0$. The expert provides $n$ different treatments $A=\{a_1,a_2,...,a_n\}$.
A treatment $a_i$ can solve the problem of $\theta_j$ if and only if $i\ge j$. Thus, $a_n$ is a panacea that fully resolves all problems. The cost of each treatment $a_i$ for the expert is $c_i$. For each treatment $a_i$, the client needs to pay the price $p_i$ to the expert. If the problem remains unsolved, the client will suffer an extra disutility $l$. We have $p_1< p_2<\cdots<p_n\ll l$. That is, the better the treatment, the higher the price. If the problem remains unsolved, the receiver will have disutility sufficiently larger than any possible prices. (Consider the medical service example, the patient may have mortal danger if not being cured.) The utility of the receiver, namely client is $u(a_i,\theta_j)=-p_i-l\cdot\mathbbm{1}_{i<j}$. Without loss of generality, we add a sufficiently large constant $c$ on the utility of the receiver to keep it positive. The receiver's utility is $u(a_i,\theta_j)=c-p_i-l\cdot\mathbbm{1}_{i<j}$. Following \citet{farukh2020inefficiency}, we assume the expert is altruistic, which means his utility doesn't contradict the client's utility. Formally, we denote the cost of each treatment as $c_i$ for all $i\in[n]$ and the sender utility is $v(a_i,\theta_j)=p_i-c_i\triangleq s_i$. We have $p_1-c_1>p_2-c_2>\cdots>p_n-c_n$.

\begin{theorem}\label{thm:credence}
In the persuasion model of the credence goods market, there exists an optimal signaling scheme that is ex-post IR. 


\end{theorem}

Technically, our result relies on a notion of \emph{greedy algorithm}. In detail, when the sender's preference over actions is $a_1\ge a_2\ge ...\ge a_n$. A natural greedy idea to compute the signaling scheme goes as follows. Given the prior $\mu_0$, we compute a posterior belief $\mu_1$ that maximizes the probability of taking action $a_1$ and subtract it from the prior. Note that this can be computed by simply running linear programming. Then compute the posterior belief that maximizes the probability of taking action $a_2,...,a_n$ respectively. In this algorithm, it will be better to treat the prior $\mu_0$ as an initial multidimensional \emph{"budget"} and after each step of the greedy algorithm, each dimension of the budget decreases. The algorithm stops until all the budget is used up. Formally, we provide the algorithm (see \cref{alg:greedy}) to solve a greedy signaling scheme. We can prove that in the credence market, the greedy signaling scheme is ex-post IR. Furthermore, we can prove that there exists a greedy signaling scheme that is optimal through the condition $p_1\le p_2\cdots p_n\ll l$. Combining two results together, we know such an optimal signaling scheme is ex-post IR.  

\begin{proof}
 Let us first prove the the optimal signaling scheme is a greedy signaling scheme. We formalize our proof through the use of mathematical induction. First, let us consider the binary action case. Then the expected sender utility can be represented as $q_1s_1+q_2s_2$ where $q_i$ is the probability of taking action i given some signaling scheme for $i=1,2$. Because $q_1+q_2=1$ and the greedy signaling scheme will maximize $q_1$, the greedy signaling scheme is optimal.

Next, assume the greedy signaling scheme is optimal when the number of actions is $n-1$. Consider the case that the number of states and actions is $n$. Note that for any prior $\mu_0$ and any $\theta_k$ for some $k\in[n]$, if $\mu_0(\theta_k)=0$, $a_k$ can not the the best response under $\mu_0$ because the receiver's utility $u$ is cyclically monotone. Therefore, it is without loss of generality to assume that $\mu_0$ does not have zero component. Otherwise we can delete $\theta_k$ and $a_k$ and reduce it to thee case of $n-1$ actions. In the following, we fix a prior $\mu_0$ that $\mu(\theta_k)>0$ for any $k\in[n]$ and an optimal signaling scheme $\pi$. We first provide a lemma of the probability of taking some action $a_i$ in the greedy signaling scheme.
\begin{lemma}\label{lem: perturbation loss}
    Consider a greedy signaling scheme $\pi^\prime$, let $q_i$ be the probability of taking action $a_i$ for any $i\in[n]$, then
    \begin{equation*}\label{eq: bind}
         q_i=\min_{i\le j\le n}\frac{\sum_{k=1}^i\mu_{i-1}(\theta_k)+\sum_{k=j+1}^n\mu_{i-1}(\theta_k)}{1-\frac{p_j-p_i}{l}},  
    \end{equation*}
where $\mu_{i-1}$ follows the definition in \cref{thm:greedy}.
\end{lemma}
\begin{proof}
    Let $x_{ij}=\frac{\sum_{k=1}^i\mu_{i-1}(\theta_k)+\sum_{k=j+1}^n\mu_{i-1}(\theta_k)}{1-\frac{p_j-p_i}{l}}$. We have $q_i=\sum_{k=1}^n\pi^\prime(a_i,\theta_k)$. On one side, consider any action $a_j$ that $i<j\le n$. According to \cref{eq: incentive}, we have
    \[
    (p_j-p_i)(\sum_{k=1}^i\pi^\prime(a_i,\theta_k)+\sum_{k=j+1}^n\pi^\prime(a_i,\theta_k)\ge (l-(p_j-p_i))\sum_{k=i+1}^j\pi^\prime(a_i,\theta_k).
    \]
Then we have $l(\sum_{k=1}^i\pi^\prime(a_i,\theta_k)+\sum_{k=j+1}^n\pi^\prime(a_i,\theta_k))\ge (l-(p_j-p_i))q_i$. So for any $i<j\le n$, we have $q_i\le x_{ij}$. We also have $q_i\le x_{ii}$ because $x_{ii}=\sum_{k=1}^n\mu_{i-1}\theta_i$. Therefore $q_i\le\min_{i\le j\le n}x_{ij}$.

On the other side, let $x_{ij^\prime}=\min_{i\le j\le n}x_{ij}$. We can build $\pi^\prime(a_i,\theta_k)$ for any $k\in[n]$ such that $q_i=x_{ij^\prime}$ in the following way. First, for $1\le k\le i$ and $j^\prime<k\le n$, let $\pi^\prime(a_i,\theta_k)=\mu_{i-1}(\theta_k)$. Then for $i<k\le j^\prime$, let $\pi^\prime(a_i,\theta_k)\le\mu_{i-1}(\theta_k)$ and 
\begin{equation}\label{eq: middle}
\sum_{i<k\le j^\prime}\pi^\prime(a_i,\theta_k)=x_{ij^\prime}-(\sum_{k=1}^i\mu_{i-1}(\theta_k)+\sum_{k=j^\prime+1}^n\mu_{i-1}(\theta_k)). 
\end{equation}
Such $\pi^\prime(a_i,\theta_k)$ exists because $x_{ij^\prime}\ge x_{ii}$ and it is easy to verify that it is a feasible solution.
\end{proof}

Next, in the optimal signaling scheme $\pi$, we will prove that $q_1$ equals the probability of taking $a_1$ computed by the greedy algorithm. Otherwise, consider in the optimal solution, $q_1$ is strictly smaller than the the probability of taking $a_1$ computed by the greedy algorithm. Let the greedy signaling scheme be $\pi^\prime$. Denote $\pi_1=(\pi(a_1,\theta_1),\cdots,\pi(a_1,\theta_k),\cdots,\pi(a_1,\theta_n)$ and $\pi_1^\prime=(\pi^\prime(a_1,\theta_1),\cdots,\pi^\prime(a_1,\theta_k),\cdots,\pi^\prime(a_1,\theta_n))$. Let $q_1^\prime=\sum_{k=1}^n\pi^\prime(a_1,\theta_k)=q_1+\epsilon$ where $\epsilon>0$. Because $\pi$ is the optimal signaling scheme and the inductive assumption, we know $\pi_2,\pi_3,\cdots,\pi_n$ can be computed by the greedy algorithm with initial "budget" $\mu_0-\pi_1$. Let $\pi^\prime$ be a new signaling scheme that $\pi_2^\prime,\pi_3^\prime,\cdots,\pi_n^\prime$ can be computed by the greedy algorithm with initial "budget" $\mu_0-\pi_1^\prime$. The gain from such perturbation is at least $s_1\epsilon$. Now let us upper bound the loss from such perturbation. According to $l\gg p_n>\cdots>p_1$ and \cref{lem: perturbation loss}, we know that for any budget $u_{i-1}$, $x_{in}=\min_{i\le j\le n}x_{ij}$. Therefore, for any $2\le i\le n$, $\pi_i$ and $\pi_i^\prime$ both bind at the incentive constraint of $a_1$ and $a_n$. For $i\ge 2$, We have $\mu_{i-1}(\theta_n)\ge\frac{p_n-p_i}{l-(p_n-p_2)}$ because $l\gg p_n$ and then there exists one feasible $\pi_i$ such that $\pi(a_i,\theta_k)=0$ for any $i<k<n$ and $\pi(a_i,\theta_n)=\frac{(p_n-p_i)\sum_{k=1}^i\pi(a_i,\theta_k)}{l-(p_n-p_i)}$. We have similar argument for $\pi^\prime(a_i,\theta_k)$ for $i\ge 1$ (Note that $\pi^\prime(a_1,\theta_k)$ is also computed by the greedy algorithm). According to \cref{lem:useup} and the above mentioned argument, we know the loss is at most $(1+\frac{(p_n-p_2)}{l-(p_n-p_2)})s_2\epsilon=\frac{1}{1-\frac{p_n-p_2}{l}}s_2\epsilon$. Therefore, the perturbation is profitable because $s_1>s_2$ and $l\gg p_n$, which means there exists an optimal signaling scheme $\pi$,  $\pi(a_1,\cdot)$ is computed by the greedy signaling scheme. The left signaling scheme is the case of $n-1$ actions and we know it should be greedy by inductive assumption.

Now we know that in the credence goods market, the optimal signaling scheme is a greedy signaling scheme. We also know that the receiver satisfies cyclical monotonicity and weak logarithmic supermodularity. According to \cref{thm:greedy}, we know the greedy signaling scheme is ex-post IR. Combining two things together, we know in the persuasion model of the credence goods market, there exists an optimal signaling scheme that is greedy and then ex-post IR. 

\end{proof}


 The game of credence market satisfies two properties we term as \emph{cyclically monotonicity} and \emph{weak logarithmic supermodularity} for the receiver's utility. The word "weak" means that the logarithmic supermodularity doesn't need to hold for the diagonal elements. Formally, we have

\begin{definition}\label{cyclical monotonicity}(cyclical monotonicity) A utility function u is cyclically monotone if for all $i,j,k\in[n]$ that $i<j$
\[
u(a_{k-1+i},\theta_k)\ge u(a_{k-1+j},\theta_k)
\]
where $a_{k-1+i}\equiv a_{(k-1+i)\mathrm{mod }\ n}$.
\end{definition}
\shuran{Explain these two definitions.}\cref{cyclical monotonicity} says that if each column's elements are arranged on a circle, they are monotone and the diagonal element $u(a_k, \theta_k)$ is the largest for all $k\in[n]$. See \cref{tab: cyc_mono} for example. Consider the second column, we have $u(a_2,\theta_2)>u(a_3,\theta_2)>u(a_4,\theta_2)>u(a_1,\theta_2)$. Other columns are similar. \jiahao{added a table to explain this definition.}

\begin{table}[!h]
\centering
\begin{tabular}{|l|l|l|l|l|}
\hline
  & $\theta_1$  & $\theta_2$  & $\theta_3$  & $\theta_4$  \\ \hline
$a_1$ & 13 & 3  & 3  & 3  \\ \hline
$a_2$ & 12 & 12 & 2  & 2  \\ \hline
$a_3$ & 11 & 11 & 11 & 1  \\ \hline
$a_4$ & 10 & 10 & 10 & 10 \\ \hline
\end{tabular}
\caption{The receiver's utility in credence goods market when $n=4$, $p_1=1$, $p_2=2$, $p_3=3$, $p_4=4$, $l=10$, and $c=14$.}
\label{tab: cyc_mono}
\end{table}

\begin{definition}\label{weak supermodular}(weak logarithmic supermodularity)
A utility function u is weakly supermodular in logarithmic form if $\frac{u(a_i,\theta_k)}{u(a_i,\theta_{k+1})}\ge\frac{u(a_j,\theta_{k})}{u(a_j,\theta_{k+1})}\quad\forall 1\le i<j,k<n,j\neq k$.  
\end{definition}
\cref{weak supermodular} is basically the supermodularity in the logarithmic form except the logarithmic supermodularity does not need to hold for the diagonal element. Intuitively, the supermodularity of receiver's utility function captures that the marginal utility from higher actions increases with the state. Then we extend the game of credence market to a general class of games. Our extension says that as long as a game satisfies cyclical monotonicity and weak supermodularity in logarithmic form, the greedy signaling scheme is ex-post IR. Therefore, the greedy signaling scheme is ex-post IR for any game in this class. As a corollary, the gap between the gap between the optimal signaling scheme and the optimal ex-post IR signaling scheme can be upper bound by the gap between the optimal signaling scheme and the greedy signaling scheme.
\begin{theorem}\label{thm:greedy}
    The greedy signaling scheme is ex-post IR if $u$ is cyclically monotone and weakly supermodular in logarithmic form. 
\end{theorem}
\shuran{Again, it is not possible for the reader to understand this paragraph. Try to extract the high-level ideas and just explain the intuition. Otherwise, this paragraph is not necessary.} 

The proof idea is as follows. Recall that the best response under prior $\mu_0$ is $a_{k^*}$. To prove the greedy signaling scheme is ex-post IR is equivalent to proving GS will stop in $k^*$ steps. Suppose \cref{alg:greedy} stops in $m$ rounds and the prior becomes $\mu_i$ after $LP_i$ for all $i\in[m]$. We will prove that for all $k>k^*$, $a_k$ can not be the best response under $\mu_i$. Therefore, \cref{alg:greedy} must stops at step $m$ when the best response under $\mu_{m-1}$ is $a_{m}$ and the algorithm will use all the “budget” to maximize the probability of taking action $a_m$ because of the greediness.

\begin{proof}

Let the left "budget" after $LP_i$ be $u_i\in\mathbb{R}^n$. Suppose \cref{alg:greedy} stops in $m$ rounds,i.e. after running $LP_1,LP_2,...,LP_m$, the prior $\mu_0$ becomes zero. The solution for $LP_i$ is $\pi_i\in\mathbb{R}^n$. We have for all $i\in[m]$,
\begin{equation}\label{eq:budget}
    \mu_0 = \sum_{j=1}^{i}\pi_j + \mu_i.
\end{equation}

We will prove that after any round $i$ of \cref{alg:greedy}, the best action under $\mu_i$ is always better than $a_{k^*}$ for the sender. This is equal to prove for any $k>k^{*}$, $a_k$ is not a bets action under $\mu_i$. Then \cref{alg:greedy} must stop at some $LP_m,m\ge k^{*}$, which means the outcome of GS is ex-post IR. In order to prove this, let us prove two following lemmas about linear programming $LP_i$. Denote the jth constraints of \cref{greedyLP:IC} in $LP_i$ as $IC_{ij}$ and the jth constraints of \cref{greedyLP:budget} in $LP_i$ as $\mathrm{budget}_{ij}$. 

\begin{lemma}\label{lem:useup}
    For any $i\in[m]$, $\mu_i(\theta_i)=0$.
\end{lemma}
\begin{proof}
Otherwise, we can increase $\pi_i(a_i,\theta_i)$ and have a higher objective value for $LP_i$ without breaking any constraints in $LP_i$ because of the cyclical monotonicity.
\end{proof}

\begin{lemma}\label{lem:bind}
    For any $i\in[m]$, linear programming $LP_i$ satisfies that at least one of $IC_{ij}$ and $\mathrm{budget}_{ij}$ holds with equality for all $j\in[n]$.
\end{lemma}
\begin{proof}
According to \cref{lem:useup}, let $(0,\cdots,0,\pi_i(\theta_{i}),\cdots,\pi_i(\theta_n))$ be the optimal solution of $LP_i$ and for $j\le i$, $\mathrm{budget}_{ij}$ binds. And for $j=n$, if $IC_{in}$ and $\mathrm{budget}_{in}$ both don't bind. We can simply increase $\pi_i(a_i,\theta_n)$ with a sufficiently small $\epsilon$ and have a higher objective value for $LP_i$ without breaking any constraints in $LP_i$ because of the cyclical monotonicity.

Next consider a fixed $k$ that $i<k<n$, $IC_{ik}$ and $\mathrm{budget}_{ik}$ both don't bind. Consider the following perturbation $(0,\cdots,\pi_i(\theta_k)+\frac{u(a_i,\theta_{k+1})}{u(a_i,\theta_k)}\epsilon,\pi_i(\theta_{k+1})-\epsilon,...,\pi_i(\theta_n))$. We will show when $IC_{ik}$ and $budget_{ik}$ both holds with inequality, there exists $\epsilon$ such that this perturbation is feasible and profitable.

\begin{itemize}
    \item \textbf{profitable} Note that if we add a constant on some column of the receiver's utility matrix $\mu$, the constraints will not change. Therefore, it is without loss of generality to assume that $u(a_i,\theta_k)\le u(a_i,\theta_{k+1})$ for all $k<n$. Otherwise, we can have $u^\prime(a_i,\theta_k)=u(a_i,\theta_k)+c_k$ and let $c_k\ll c_{k+1}$ for all $k<n$ and replace $u$ by $u^\prime$. Then we have $\frac{u(a_i,\theta_{k+1})}{u(a_i,\theta_k)}\ge 1$, $\sum_{j=1}^n\pi_i(a_i,\theta_j)$ will increase. The perturbation is profitable.
    
    \item \textbf{feasible}
    For sufficiently small $\epsilon$, $\mathrm{budget}_{ik}$ and $IC_{ik}$ hold. For any $j\neq k$, $IC_{ij}$ is
    \begin{equation*}
    \sum_k(u(a_j,\theta_k)-u(a_i,\theta_k))\pi(a_i,\theta_k)\le 0 
    \end{equation*}
    The perturbation will lead to a change $\Delta$ on the left side:
    \begin{equation*}
        \begin{aligned}
            \Delta &= -(u(a_j,\theta_{k+1})-u(a_i,\theta_{k+1}))\epsilon+(u(a_j,\theta_{k})-u(a_i,\theta_{k}))\frac{u(a_i,\theta_{k+1})}{u(a_i,\theta_k)}\epsilon\\
            &=\frac{\epsilon}{u(a_i,\theta_k)}(u(a_j,\theta_k)u(a_i,\theta_{k+1})-u(a_j,\theta_{k+1})u(a_i,\theta_k))\le 0
        \end{aligned}
    \end{equation*}
    The above inequality holds because of the weak logarithmic supermodularity.
\end{itemize}   
\end{proof}

Next, we will use induction to prove that the best response under $\mu_i$ is always better than $a_{k^*}$. Consider the base case $i=1$. Let $A=\{k|k>k^*, \mu_1(\theta_k)>0\},B=\{k|k>k^*, \mu_1(\theta_k)=0\}$. Because at least one of $IC_{1j}$ and $budget_{1j}$ holds with equality, for any $k\in A$, $u_1(\theta_k)>0$ implies that under $\pi_1$, $a_k$ is also a best action. Let $u_i$ be the ith row of the receiver's utility matrix. Then we have 
\begin{equation*}
    (u_{k^*}-u_k)\cdot\pi_1=(u_{k^*}-u_1)\cdot\pi_1\le 0.
\end{equation*}
Recall that the best response under $\mu_0$ is $a_{k^*}$,
\begin{equation*}
    (u_{k^*}-u_k)\cdot\mu_0\ge 0.
\end{equation*}
Combining the above two inequalities together, we have
\begin{equation*}
    (u_{k^*}-u_k)\cdot(\mu-\pi_1)=(u_{k^*}-u_k)\cdot\mu_1\ge 0,
\end{equation*}
which means under $\mu_1$, $a_{k^*}$ is a better action than $a_k$ for receiver.

Next, consider $k\in B$. For such a fixed $k$, $\mu_k\cdot\mu_1\le\mu_{k-1}\cdot\mu_1$ because $u$ is cyclically monotone. Therefore, $a_k$ is not the best response under $\mu_1$. Combining the results for $k\in A$ and $k\in B$ together, we know that for any $k>k^{*}$, $a_k$ is not the best response under $\mu_1$.

Now suppose that for some $i$, $a_k$ is not the best response under $\mu_i$. Consider $\mu_{i+1}$. We have $\mu_{i}(\theta_j)=0$ for all $j\le i$ according to \cref{lem:useup} and \cref{eq:budget}. Since $\mu_{i}$ is the input for $LP_{i+1}$, it is sufficient to consider $a_{i+1},...a_n$ and $\theta_{i+1},...,\theta_n$ in $LP_{i+1}$. That is equal to considering submatrix $u[i+1:n;i+1:n]$. We can similarly prove that at least one of $IC_{i+1,j}$ and $\mathrm{budget}_{i+1,j}$ binds. Then the left budget $\mu_{i+1}$ will also lead to a best response that is no worse than $a_{k^*}$ for the sender. 
\end{proof}


\begin{corollary}\label{cor:greedy}
    If a game $(A,\Theta,v,u)$ satisfies that $\InAbs{A}=\InAbs{\Theta}$ and $u$ is cyclically monotone and weakly supermodular in logarithmic form, for any prior $\mu_0$, we have
    \[
    V(\mu_0)-V_{\textsc{ex-post}}(\mu_0)\le V(\mu_0)-V_{\textsc{greedy}}(\mu_0)
    \]
    where $V_{\textsc{greedy}}(\mu_0)$ is the sender's utility under the optimal greedy signaling scheme.
\end{corollary}
In algorithm analysis, the greedy algorithm by nature has a close link with the optimal solution. We believe that \cref{cor:greedy} will be useful in finding more games where there is no gap between the optimal signaling schemes with and without the ex-post IR constraint.
\begin{algorithm}
    \caption{GS: greedy signaling scheme}
    \label{alg:greedy}

    \KwIn{prior $\mu_0=(\mu_0(\theta_1),...,\mu_0(\theta_n))$}
    $\pi\gets 0$\;
    $\mu\gets \mu_0$\;
    \While{$\mu\neq 0$}{
    solve $\pi(a_i,\cdot)$ by $LP_i$ so that the probability of taking action $a_i$ is maximized\;
    $\mu\gets\mu - \pi(a_i,\cdot)$\;
    }
    \KwOut{a feasible outcome $\pi(a,\theta)$}

    where $LP_i$ is
\begin{equation*}
    \max\sum_k \pi(a_i,\theta_k)
\end{equation*}
subject to 
\begin{equation}\label{greedyLP:IC}
    \sum_k(u(a_j,\theta_k)-u(a_i,\theta_k))\pi(a_i,\theta_k)\le 0 \quad\forall j\in[n]
\end{equation}
\begin{equation}\label{greedyLP:budget}
    \pi(a_i,\theta_k)\le\mu(\theta_k)\quad\forall k\in[n]
\end{equation}
\end{algorithm}

\section{Missing proofs from \texorpdfstring{\MakeLowercase{\cref{section:gap}}}{Section 4}}
\subsection{Missing Proofs from \texorpdfstring{\cref{subsection:binary}}{Section 4.1}}
\textbf{Proof of \cref{lem:quasiconcave}}
\begin{proof}
In the binary state case, we slightly abuse notation by using $\mu$ to represent $\mu(\theta_1)$. Consider any possible linear combination $\mu=\lambda\mu_1+(1-\lambda)\mu_2$ that $\lambda$ is a real number in $[0,1]$. According to the definition of quasiconcave closure, we have
\[
\overline{V}(\mu)\ge\min\{\overline{V}(\mu_1),\overline{V}(\mu_2)\}\ge \min\{\hat{v}(\mu_1),\hat{v}(\mu_2)\}.
\]

On one hand, we can assume $\mu_1\le\mu\le\mu_2$ without loss of generality. Combining the inequalities led by all linear combinations, we have $\overline{V}(\mu)\ge\min\{\max_{\mu_1:\mu_1\le\mu}\hat{v}(\mu_1),\max_{\mu_2:\mu_2>\mu}\hat{v}(\mu_2)\}$
On the other hand, if a closure of $\hat{v}$ satisfies its value of $\mu$ is $\min\{\max_{\mu_1:\mu_1\le\mu}\hat{v}(\mu_1),\max_{\mu_2:\mu_2>\mu}\hat{v}(\mu_2)\}$ for all $\mu\in[0,1]$, it is indeed quasiconcave. Therefore, we have 
\begin{equation}\label{eq:quasiconcave}
   \overline{V}(\mu)=\min\{\max_{\mu_1:\mu_1\le\mu}\hat{v}(\mu_1),\max_{\mu_2:\mu_2>\mu}\hat{v}(\mu_2)\}. 
\end{equation}
Then $\overline{V}$ is a linear function because $\hat{v}$ is a linear piecewise function, which corresponds to a partition $\overline{P}=(\{J_j\}_{j=1}^m,\{\beta_j\}_{j=0}^m)$ of $[0,1]$.

Now let us prove that there exists $j^*\in[m]$ such that $I_{i^*}\subset J_{j^*}$. According to $a^*(S_{i^*})=a_1$ and \cref{sorted}, there exists $i^*$ such that $\hat{v}(\mu)>\hat{v}(\mu^\prime)$ for all $\mu\in(\alpha_{i^*-1},\alpha_{i^*})$ and $\mu^\prime\notin(\alpha_{i^*-1},\alpha_{i^*}).$ Then $\overline{V}(\mu)=\hat{v}(\mu)$ for all $\mu\in(\alpha_{i^*-1},\alpha_{i^*})$, which means $\overline{V}$ is continuous in $I_{i^*}$. Therefore, there exists an interval in the partition $\overline{P}$ contains interval $I_{i^*}$. Next, consider the part of $\overline{V}$ left to interval $J_{j^*}$, i.e. say $\mu<\beta_{j^*-1}$. We have $\max_{\mu_2:\mu_2>\mu}\hat{v}(\mu_2)=\max\{\hat{v}(\beta_{j^*-1}),\hat{v}(\beta_{j^*})\}$. Because $\max\{\hat{v}(\beta_{j^*-1}),\hat{v}(\beta_{j^*})\}>\max\limits_{\mu_1:\mu_1\le\mu}\hat{v}(\mu_1)$ and \cref{eq:quasiconcave}, we have $\overline{V}(\mu)=\max\limits_{\mu_1:\mu_1\le\mu}\hat{v}(\mu_1)$ for all $\mu<\beta_{j^*-1}$. Now it is not hard to see from the above equation that $\overline{V}$ is increasing on $[0,\beta_{j^*-1})$. So the quasiconcave closure $\overline{V}$ is increasing w.r.t intervals $J_j$ for all $j<j^*$. As for the part right to the $J_{j^*}$, we can similarly have $\overline{V}=\max\limits_{\mu_2:\mu_2\ge\mu}\hat{v}(\mu_2)$ for all $\mu>\beta_{j^*}$. Then $\overline{V}$ is decreasing on $(\beta_{j^*},1]$. We have the quasiconcave closure $\overline{V}$ is decreasing w.r.t intervals $J_j$ for all $j>j^*$. 
\end{proof}

The following lemma characterized the geometric property of the concave closure $V$ of the sender's expected utility function $\hat{v}$. Given the partition $\hat{P}=(\{I_i\}_{i=1}^n,\{\alpha_i\}_{i=0}^n)$, $\hat{v}$ can be divided into $n$ segments $S_1,\cdots,S_n$ such that each segment $S_i=\{(x,\hat{v}(x))\mid x\in I_i\}$ is a set of points for all $i\in[n]$. Denote the left endpoint and right endpoint of segment $S_i$ be $L_i$ and $R_i$ respectively. Let the corresponding best response when $\mu\in I_i$ be $a^*(S_i)$. In the following analysis, without loss of generality, we assume that there exists a unique $i^*$ such that $a^*(S_{i^*})=a_1$. 
\begin{lemma}\label{lem:concave}
The concave closure V is a connecting segment 
\[L_{i_1}L_{i_2}...L_{i_{k-1}}L_{i^*}R_{i^*}R_{i_k}...R_{i_t}\] where $i_1<i_2,..<i_{k-1}<i^*<i_k<...<i_t$. 
\end{lemma}
\begin{proof}
In the binary state case, we slightly abuse notation by using $\mu$ to represent $\mu(\theta_1)$. Consider any possible linear combination $\mu=\lambda\mu_1+(1-\lambda)\mu_2$ that $\lambda$ is a real number in $[0,1]$. According to the definition of concave closure, we have
\[
V(\mu)=\max_{\lambda,\mu_1,\mu_2\in[0,1]:\mu=\lambda\mu_1+(1-\lambda)\mu_2}\lambda\hat{v}(\mu_1)+(1-\lambda)\hat{v}(\mu_2).
\]
So $V(\mu)=\hat{v}(\mu)$ for all $\mu\in[\alpha_{i^*-1},\alpha_{i^*}]$, which means $S_{i^*}$ is one part of the concave closure $V$. According to the definition of concave closure, $V$ is a connecting segment whose vertexes are in the set of $\{L_i\}_{i=1}^n\cup\{R_i\}_{i=1}^n$. Next, we use the induction method to prove that the vertexes left to the $L_{i^*}$ must be some left endpoints of the segments of the sender's expected utility function $\hat{v}$. Say there are $k-1$ vertexes left to $L_{i^*}$ and we start our induction from the left adjacent vertex of $L_{i^*}$, which is chosen from some segment $S_{i_{k-1}}$. Now we only need to prove that this vertex is $L_{i_{k-1}}$ rather than $R_{i_{k-1}}$. This is equal to prove that the slope of $L_{i_{k-1}}L_{i^*}$ is smaller than the slope of $R_{i_{k-1}}L_{i^*}$.

The ordinates of $L_{i_{k-1}},R_{i_{k-1}},L_{i^*}$ are $(\alpha_{i_{k-1}-1},\hat{v}(\alpha_{i_{k-1}-1}))$, $(\alpha_{i_{k-1}},\hat{v}(\alpha_{i_{k-1}}))$, $(\alpha_{i^*-1},\hat{v}(\alpha_{i^*-1}))$ respectively. For ease of notation, denote $a^*(S_{i_{k-1}})$ as $a_x$ and recall that $a^*(S_{i^*})=a_1$. We have $a_1\ge a_x$.

The slope of $L_{i_{k-1}}L_{i^*}$ is
\[
k_1=\frac{\hat{v}(\alpha_{i^*-1})-\hat{v}(\alpha_{i_{k-1}-1})}{\alpha_{i^*-1}-\alpha_{i_{k-1}-1}}.
\]
The slope of $R_{i_{k-1}}L_{i^*}$ is 
\[
k_2=\frac{\hat{v}(\alpha_{i^*-1})-\hat{v}(\alpha_{i_{k-1}})}{\alpha_{i^*-1}-\alpha_{i_{k-1}}}.
\]
We have
\[
\begin{aligned}
 k_2-k_1&=\frac{\hat{v}(\alpha_{i^*-1})-\hat{v}(\alpha_{i_{k-1}})}{\alpha_{i^*-1}-\alpha_{i_{k-1}}}-\frac{\hat{v}(\alpha_{i^*-1})-\hat{v}(\alpha_{i_{k-1}-1})}{\alpha_{i^*-1}-\alpha_{i_{k-1}-1}}\\
 &=\frac{\alpha_{i_{k-1}}-\alpha_{i_{k-1}-1}}{\alpha_{i^*-1}-\alpha_{i_{k-1}-1}}\InBrackets{\frac{\hat{v}(\alpha_{i^*-1})-\hat{v}(\alpha_{i_{k-1}})}{\alpha_{i^*-1}-\alpha_{i_{k-1}}}-\frac{\hat{v}(\alpha_{i_{k-1}})-\hat{v}(\alpha_{i_{k-1}-1})}{\alpha_{i_{k-1}}-\alpha_{i_{k-1}-1}}}\\
 &=\frac{\alpha_{i_{k-1}}-\alpha_{i_{k-1}-1}}{\alpha_{i^*-1}-\alpha_{i_{k-1}-1}}\Big[\frac{v(a_1,\theta_1)\alpha_{i^*-1}+v(a_1,\theta_2)(1-\alpha_{i^*-1})-v(a_x,\theta_1)\alpha_{i_{k-1}}-v(a_x,\theta_2)(1-\alpha_{i_{k-1}})}{\alpha_{i^*-1}-\alpha_{i_{k-1}}}\\
 &-(v(a_x,\theta_1)-v(a_x,\theta_2))\Big]\\
 &=\frac{\alpha_{i_{k-1}}-\alpha_{i_{k-1}-1}}{\alpha_{i^*-1}-\alpha_{i_{k-1}-1}}\cdot\frac{(v(a_1,\theta_1)-v(a_x,\theta_1))\alpha_{i^*-1}+(v(a_1,\theta_2)-v(a_x,\theta_2))(1-\alpha_{i^*-1})}{\alpha_{i^*-1}-\alpha_{i_{k-1}}}\ge 0\\
\end{aligned}
\]
where the inequality holds because $\alpha_{i^*-1}>\alpha_{i_{k-1}}>\alpha_{i_{k-1}-1}$ and \cref{sorted}.

Then consider the left adjacent vertex of $L_{k-1}$, which is chosen from some segment $S_{i_{k-1}-1}$. We can similarly prove that this vertex is $L_{i_{k-1}-1}$ rather than $R_{i_{k-1}-1}$. By induction, we know that the part of concave closure left to the segment $S_{i^*}$ is a connecting segment $L_{i_1}L_{i_2}\cdots L_{i_{k-1}}$. As for the part right to the segment $S_{i^*}$, we can similarly use the induction method to prove that the vertexes right to $R_{i^*}$ are the right endpoints of some segments of the sender's expected utility function $\hat{v}$. Combine two directions together, we have the concave closure $V$ is a connecting segment $L_{i_1}L_{i_2}...L_{i_{k-1}}L_{i^*}R_{i^*}R_{i_k}...R_{i_t}$ where $i_1<i_2,..<i_{k-1}<i^*<i_k<...<i_t$. 
\end{proof}

Now we are ready to prove \cref{thm:n*2}.

\noindent\textbf{Proof of \cref{thm:n*2}}
\begin{proof}
In this proof, we assume $a_1>a_2\ge\cdots\ge a_n$. In other words, there exists a unique $i^*$ such that $a^*(S_{i^*})=a_1$. With this assumption, we know $J_{j^*}=I_{i^*}$, $\beta_{j^*-1}=\alpha_{i^*-1}$ and $\beta_{j^*}=\alpha_{i^*}$. This is without loss of generality because when there are $r>1$ indexes $i^*_1<\cdots i^*_r$ corresponding to $a_1$, we can replace $L_{i^*}R_{i^*}$ by $L_{i^*_1}R_{i^*_r}$, $\beta_{j^*-1}$ by $\alpha_{i^*_1}-1$ and $\beta_{j^*}$ by $\alpha_{i^*_r}$. According to \cref{lem:concave}, consider one segment $L_{i_j}L_{i_{j+1}}$ of the concave closure on the left of $S_{i^*}$ for some $j\in\{1,2,\cdots,k-1\}$. Because $a_1$ strictly better than other actions and the concavity, we know $a^*(S_{i_{j+1}})>a^*(S_{i_j})$. If there exists an index $i^\prime$ that $i_j<i^\prime<i_{j+1}$ and $a^*(S_{i_{j+1}})\ge a^*(S_{i^\prime})>a^*(S_{i_j})$, for any posterior belief $\mu$ that $a^*(\mu)=a^*(S_{i^\prime})$, the sender will send a signal that the receiver's best response is a worse action $a^*(S_{i_j})$, which is not ex-post IR. In other words, for any prior $\mu_0$ that $\mu_0(\theta_1)<\alpha_{i^*-1}$, the optimal signal scheme is ex-post IR if and only if for any $j\in\{1,2,\cdots,k-1\}$, $L_{i_{j+1}}$ is the left endpoint of its right closet segment whose corresponding action is better than $a^*(S_{i_j})$ for the sender. Similarly, for any prior $\mu_0$ that $\mu_0(\theta_1)>\alpha_{i^*}$, the optimal signal scheme is ex-post IR if and only if for any $j\in\{k,k+1,\cdots,m\}$, $R_{i_j}$ is the right endpoint of its right closet segment whose corresponding action is worse than $a^*(S_{i_{j-1}})$ for the sender.

In other words, when the optimal signaling scheme is ex-post IR, $L_{i_1}L_{i_2}\cdots L_{i_{k-1}}$ can be formed in the following way. We start from $L_{i_1}=L_1$ and find the left endpoints $L_{i_2}$ that is the closest left endpoint to $L_{i_1}$ which satisfies $a^*(S_{i_2})>a^*(S_{i_1})$. Then we start from $L_{i_2}$ and find $L_{i_3}$ in a similar way and so on. Now we prove that the $L_{i_1}, L_{i_2},\cdots, L_{i_{k-1}}$ formed in this way are exactly all the left endpoints of the segments of the quasiconcave closure $\overline{V}$ which are on the left of $S_{i^*}$. We use the induction method. For the base case, $L_{i_1}=L_1$ is certainly the first left endpoint of $\overline{V}$. Now consider the left endpoint $L_{i_2}=(\alpha_{i_2-1},\hat{v}(\alpha_{i_2-1}))$. According to \cref{lem:quasiconcave}, for the left part, i.e. $\mu<\alpha_{i^*-1}$, we have
\[
\overline{V}(\mu)=\max\limits_{\mu_1:\mu_1\le\mu}\hat{v}(\mu_1).
\]

Note $L_{i_2}$ is the closest left endpoints that $a^*(S_{i_2})> a^*(S_{i_1})$. We have for any $\mu\in(\alpha_{i_1},\alpha_{i_2-1})$, $a^*(\mu)<a^*(S_{i_1})$. Then $\hat{v}(\mu)\le\max\{\hat{v}(0),\hat{v}(\alpha_{i_1})\}$. So $\overline{V}(\mu)=\max\{\hat{v}(0),\hat{v}(\alpha_{i_1})\}$ for all $\mu\in(\alpha_{i_1},\alpha_{i_2-1})$ and $\overline{V}(\alpha_{i_2-1})=\hat{v}(\alpha_{i_2-1})>\max\{\hat{v}(0),\hat{v}(\alpha_{i_1})\}$. Therefore, no matter $\overline{V}(\mu)=\hat{v}(0)$ or $\overline{V}(\mu)=\hat{v}(\alpha_{i_1})$, we always have $\overline{V}$ is continuous on $[0,\alpha_{i_2-1})$ and $\mu=\alpha_{i_2-1}$ is a discontinuity point. So $L_{i_2}$ is the second left endpoint of the segments of the quasiconcave closure $\overline{V}$. Because $a^*(S_{i_2})>a^*(\mu)$ for all $\mu\in[0,\alpha_{i_2-1})$, we have $\overline{V(\mu})=\max\limits_{\mu_1:\alpha_{i_2-1}\le\mu_1\le\mu}\hat{v}(\mu_1)$ for all $\mu\in[\alpha_{i_2-1},\alpha_{i^*-1})$, we can consider $L_{i_2}$ as the new start point and use the same argument to prove that $L_{i_3}$ is the next left endpoint of the segments of the quasiconcave closure $\overline{V}$. By induction, we know $L_{i_1}, L_{i_2},\cdots, L_{i_{k-1}}$ have the desired property. As for the right part of $S_{i^*}$, we can similarly prove that $R_{i_k},\cdots, R_{i_m}$ are exactly the right endpoints of the segments of $\overline{V}$ that are on the right of $S_{i^*}$.

Now conclude all the arguments together, on one side we know that if the optimal signaling scheme is ex-post IR for any prior, the left endpoints of the concave closure $V$ on the left of $L_{i^*}$ match the left endpoints of the quasiconcave closure $\overline{V}$ on the left of $L_{i^*}$. The right part is similar. According to the definition of concave closure $V$ and smoothed quasiconcave closure, we know that the smoothed quasiconcave closure is concave and equals to $V$. On the other side, if the smoothed quasiconcave closure is concave, according to $\cref{lem:concave}$, it is easy to see that the smoothed quasiconcave closure is actually the concave closure of the sender's expected utility $\hat{v}$.
\end{proof}

\subsection{Missing Proofs from \texorpdfstring{\cref{subsection:bilateral}}{Section 4.2}}
\textbf{Proof of \cref{lem:upper}}
\begin{proof}
Let's use induction on the number of the support $\InAbs{S}$. First, consider $\InAbs{S}=1$ and $S=\{k\}$ for some $k$. Let $\mu_S=\delta(\theta_k)$ where $\delta(\theta_k)$ means a measure put all mass on $\theta_k$ (Dirac delta). Then the receiver's best response is $a_k$ because $u(a_k,\theta_k)\ge u(a_i,\theta_k)$ for all $i\in[n]$.

Next, assume the lemma holds for $\InAbs{S}=m$. Consider the case that $\InAbs{S}=m+1$. Let $S=\{i_1,i_2,\cdots,i_{m+1}\}$ and $S^\prime=S\setminus\{i_{1}\}$. By inductive assumption, there exists a posterior belief $\mu_{S^\prime}$ that $\sum_{i\in S^\prime}\mu_{S^\prime}(\theta_i)=1$ and the receiver is indifferent from $a_{i_2},\cdots, a_{i_m}, a_{i_{m+1}}$. In other words, for all $i_j\in S^\prime$, there exists a constant c that $\sum_{k=1}^n u(a_{i_j},\theta_k)\mu_{S^\prime}(\theta_k)=c.$ 

Consider a posterior belief $\mu_t=t\delta(\theta_{i_1})+(1-t)\mu_{S^\prime}$. We have for all $i_j\in S^\prime$, $\sum_{k=1}^n u(a_{i_j,\theta_k})\mu_t(\theta_k)=t\cdot u(a_{i_j},\theta_{i_1})+(1-t)c=(1-t)c.$
because $i_j>i_1$ and $u$ is upper-triangular. Therefore, for any $t\in[0,1] $, under posterior belief $\mu_t$, the receiver is still indifferent from $a_{i_2},\cdots, a_{i_m}, a_{i_m+1}$. And the utility gap of taking action $a_{i_1}$ and $a_{i_j}$ is 
\[
g(t)=\sum_{k=1}^n (u(a_{i_1,\theta_k})-u(a_{i_j,\theta_k}))\mu_t(\theta_k)=t\cdot u(a_{i_1},\theta_{i_1})+(1-t)(\sum_{l=2}^{m+1}(u(a_{i_1},\theta_{i_l})-(a_{i_j},\theta_{i_l}))).
\]
We have $g(1)=u(a_{i_1},\theta_{i_1})>0$ and  $g(0)=-\sum_{l=2}^{m+1}(u(a_{i_1},\theta_{i_l})-(a_{i_j},\theta_{i_l}))<0$. Because g(t) is a linear function w.r.t $t$, there exist some $t_0$ that $g(t_0)=0$, which means under posterior $\mu_{S}=\mu_{t_0}$, the receiver is indifferent from action $a_{i_1},a_{i_2},\cdots,a_{i_{m+1}}$. At last, notice that $u$ is upper-triangular and the non-zero part of each column is increasing, if there is a belief $\mu$ doesn't have support on some $\theta_k$, $a_k$ is dominated by $a_{k+1}$ under $\mu$. So action $a_k$ must not be the best response. Now $\mu_{S}$ only have support on $\theta_{i_1},\theta_{i_2},\cdots,\theta_{i_{m+1}}$, $a_k$ can be the best response if $k\notin S$. Therefore under belief $\mu_S$, the receiver's best response is $a_k$ that $k\in S$. This completes the induction.
\end{proof}

\noindent\textbf{Proof of \cref{extendBBM}}\label{proof:extendBBM}
\begin{proof}
    According to the original Bayesian persuasion paper by \cite{kamenica2011bayesian}, a signaling scheme $\phi$ is equal to a Bayes-plausible distribution $\tau$ of posteriors. The Bayes-plausibility means that the prior $\mu_0=\sum_{\mu\in\mathrm{supp}(\tau)}\tau(\mu)\mu$. In other words, the prior is a linear combination of the posteriors. In order to prove that the optimal signaling scheme is ex-post IR for all trade-like games, we will construct a linear combination of posteriors for the prior step-by-step and prove that such a signaling scheme is ex-post IR and optimal.

    At step 0, let $S_0=\{1,2,\cdots,n\}$. According to \cref{lem:upper}, there is a posterior $\mu_{S_0}$ that the receiver's best response under it is $a_k$ for all $k\in S_0$. We let the probability of $\mu_{S_0}$ be $p_0=\min_{k=1}^n\frac{\mu_0(\theta_k)}{\mu_{S_0}(\theta_k)}$. In other words, we subtract a certain proportion of $\mu_{S_0}$ from $\mu_0$ until some supports of $\mu_0$ become zero. Let $\mu_1=\mu_0-p_0\mu_{S_0}$. Now suppose $\mu_t$ is computed after step $t-1$ for some $t\ge1$. Consider the step $t$, let $S_t=\{k: \theta_k\in\mathrm{supp}(\mu_t)\}$. We still have a posterior $\mu_{S_t}$ that the receiver's best response under it is $a_k$ for all $k\in S_t$. Let the probability of $\mu_{S_t}$ be $p_t=\min_{k=1}^n\frac{\mu_t(\theta_k)}{\mu_{S_t}(\theta_k)}$. And we have $\mu_{t+1}=\mu_0-p_t\mu_{S_t}$. The whole process stops if $\mu_{t+1}$ is a zero vector. Note that the process will stop in at most $n$ steps because, after each step, at least one component of the vector will become zero. Assume the process stops at some step $m$, we have 
    \begin{equation}\label{eq:divide}
        \mu_0=\sum_{t=0}^mp_t\mu_{S_t}. 
    \end{equation}
    Recall $a^*(\mu_0)=a_{k^*}$. Next, we will prove that $k^*\in S_t$ for all $t\in[m]$. Consider some index $i\in S_m$. For any $t\in[m]$, we have $i\in S_m$. Therefore, for any $t\in[m]$, $a_i$ is one best response under posterior $\mu_{S_t}$. We have $a_i$ is one best response under $\mu_0$ according to \cref{eq:divide}. This means $k^*\in S_m$ and hence $k^*\in S_t$ for all $t\in[m]$. In other words, $a_{k^*}$ is one best action under $\mu_{S_t}$ for all $t\in[m]$

    So far we have constructed a linear combination of the prior, i.e. we have constructed a signaling scheme. Next, we will prove such a signaling scheme is ex-post IR and optimal. Because we have proved that $a_{k^*}$ is one best action for all the posteriors in this signaling scheme and the receiver breaks ties in favor of the sender, the receiver will always choose an action that is better than the best response when no communication. In other words, this signaling scheme is ex-post IR. Now the last thing is to prove that this signaling scheme is optimal. On one side, because the sum of the non-zero part corresponding columns of the sender utility and the receiver utility is a constant $c_k$ for all $k\in[n]$, we have the maximum total utility of the sender and receiver is $\sum_{k=1}^nc_k\mu_0(\theta_k)$. In our signaling scheme, for any $k\in n$, as long as $k\in S_t$ for some $t\in[m]$, we have the best response under $\mu_{S_t}$ is $a_{\min S_t}<a_k$. Therefore, the trade always happens in our signaling scheme and hence the total utility is the optimal utility $\sum_{k=1}^nc_k\mu_0(\theta_k)$. On the other side, the receiver's utility in our signaling scheme equals the utility when there is no communication because $k^*\in S_t$ for all $t\in[m]$. It is well known that such a utility is the worst utility the receiver can obtain in any signaling scheme. Combining two results together, our signaling scheme is optimal for the sender.  
\end{proof}

\section{Credible persuasion}\label{sec:credible}
In this section, we give a brief introduction of credible persuasion. The interaction between the sender and the receiver goes on as follows:
\begin{enumerate}    
    \item sender chooses a signaling scheme $\phi$. Different from the Bayesian persuasion, the receiver cannot observe the signaling scheme;
    \item sender observes a realized state of nature $\theta\sim\mu_0$;
    \item sender draws a signal $s\in S$ according to the distribution $\phi_{\theta}$ and sends it to the receiver;
    \item receiver observes the signal $s$. Different from the Bayesian persuasion, the receiver can also observe the marginal distribution of signals induced by $\phi$. 
    
    \item The receiver chooses an action based on both the signal $s$ and the distribution of signals.
\end{enumerate}

Their notion of credibility is captured by the restriction that the receiver can only see the distribution of the signal induced by a signaling scheme $\phi$ instead of the whole signaling scheme. The receiver can not distinguish the signaling schemes that induce the same distribution of signals as $\phi$. Denote the set of such signaling schemes $D(\phi)$. They require when the sender chooses a signaling scheme $\phi$, it has no profitable deviation in $D(\phi)$.

\section{Missing Proofs from \texorpdfstring{\MakeLowercase{\cref{section:compare}}}{Section 5}}
\textbf{Proof of \cref{thm:compare}}
\begin{proof}
Let us first prove that ex-post IR Bayesian persuasion and credible persuasion are indeed between the Bayesian persuasion and cheap talk. Assume that in cheap talk, the signal/message set is $M$. The sender's strategy set is $\Sigma^S:\Theta\rightarrow\Delta(M)$ and the receiver's strategy set is $\Sigma^R:M\rightarrow \Delta(A)$. 

By definition $V_{\textsc{BP}}\succeq V_{\textsc{credible}}$. To see why $V_{\textsc{credible}}\succeq V_{\textsc{cheap}}$, consider any nash equilibrium of cheap talk $(\Sigma^{S*},\Sigma^{R*})$. It induces a distribution $p^{*}$ of message: $p^{*}(m)=\sum_{\theta}u_0(\theta)\Sigma^{S*}(m|\theta)$. The sender can partially commit $p^{*}$ and the credibility guarantees that the optimal sender utility is better than the sender utility of cheap talk equilibrium. 

By definition $V_{\textsc{BP}}\succeq V_{\textsc{cheap}}$. Now suppose actions are ordered, we will prove $V_{\textsc{expost}}\succeq V_{\textsc{cheap}}$. If for every message $m$, the receiver responds to a (mixed) action $a<a^*(\mu_0)$, then the sender's payoff is worse than no communication, let alone ex-post IR Bayesian persuasion. If there exists message $m$, the receiver response is a (mixed) action $a\ge a^*(\mu_0)$, then any cheap talk equilibrium is ex-post IR. This implies any payoff achieved by cheap talk can be achieved by ex-post IR Bayesian persuasion.

To prove the rest part of \cref{thm:compare}, we need the following lemmas. Through these lemmas, we can construct examples showing that $V_{\textsc{credible}}$ and $V_{\textsc{expost}}$ have no dominant relationship.

\begin{lemma}(\cite{lin2022credible})
    If $v(\theta,a)=f(\theta)+g(a)$ (additively separable), for any prior $\mu_0$, $V_{\textsc{BP}}(\mu_0)=V_{\textsc{credible}}(\mu_0)$.
\end{lemma}

\begin{lemma}(\cite{lin2022credible})
When $v$ is supermodular and $u$ is submodular, for any prior $\mu_0$, $V_{\textsc{credible}}(\mu_0)$ equals the sender's utility when there is no communication between the sender and the receiver.
\end{lemma}

\begin{lemma}
     When $|A|=2$, for any prior $\mu_0$, $V_{\textsc{BP}}(\mu_0)=V_{\textsc{expost}}(\mu_0)$.
\end{lemma}
The proof of the last lemma is straightforward: $A=\{a_1,a_2\},a_1>a_2$. If $a_0=a_2$, then any signaling scheme is ex-post IR. If $a_0=a_1$, no disclosure is optimal and ex-post IR. Next, we present one example that the optimal sender utility of credible persuasion is higher than the optimal sender utility of ex-post IR Bayesian persuasion and another example that is the opposite.

\begin{example}{$V_{\textsc{credible}}(\mu_0)\ge V_{\textsc{expost}}(\mu_0)$}
\begin{table}[h]
\centering
\begin{tabular}{|l|l|l|}
\hline
  & $\theta_1$ & $\theta_2$ \\ \hline
$a_1$ & 0 & 0 \\ \hline
$a_2$ & 1/2 & 1/2 \\ \hline
$a_3$ & 4 & 4 \\ \hline
\end{tabular}
\begin{tabular}{|l|l|l|}
\hline
  & $\theta_1$ & $\theta_2$ \\ \hline
$a_1$ & 0 & -16 \\ \hline
$a_2$ & -4 & -4 \\ \hline
$a_3$ & -16 & 0 \\ \hline
\end{tabular}
\end{table} 
$v$ is additively separable, for any prior $\mu_0$, $V_{\textsc{BP}}=V_{\textsc{credible}}$. By simple computation, we have  $V_{\textsc{BP}}\ge V_{\textsc{expost}}$. Therefore $V_{\textsc{credible}}(\mu_0)\ge V_{\textsc{expost}}(\mu_0)$.
\end{example}

\begin{example}{$V_{\textsc{expost}}(\mu_0)\ge V_{\textsc{credible}}(\mu_0)$}
\begin{table}[h]
\centering
\begin{tabular}{|l|l|l|}
\hline
  & $\theta_1$ & $\theta_2$ \\ \hline
$a_1$ & 1 & 1 \\ \hline
$a_2$ & 2 & 3 \\ \hline
\end{tabular}
\begin{tabular}{|l|l|l|}
\hline
  & $\theta_1$ & $\theta_2$ \\ \hline
$a_1$ & 1 & 1 \\ \hline
$a_2$ & 2 & -1 \\ \hline
\end{tabular}
\end{table}
$v$ is supermodular, $u$ is submodular and $|A|=2$. Ex-post IR Bayesian persuasion has no gap to Bayesian persuasion, but credible persuasion equals to no communication. So for any $\mu_0$, $V_{\textsc{expost}}=V_{\textsc{BP}}>V_{\textsc{credible}}$.
\end{example}

\end{proof}

\end{document}